\documentclass[a4paper,11pt,dvipsnames]{article}
\usepackage{jheppub} % for details on the use of the package, please see the JINST-author-manual
\usepackage{lineno}

\usepackage{orcidlink}
\usepackage{hyperref}

\usepackage{dsfont}
\usepackage{mathrsfs}

\usepackage{algorithm}
\usepackage[noend]{algpseudocode}

\usepackage{amsthm}
\numberwithin{equation}{section}
\theoremstyle{plain}
\newtheorem{theorem}{Theorem}[section]

\theoremstyle{definition}
\newtheorem{definition}[theorem]{Definition}

\theoremstyle{remark}

\usepackage{booktabs}
\usepackage{siunitx}
\usepackage{multirow}
\usepackage{colortbl}
\renewcommand{\arraystretch}{1.2}
\usepackage{xcolor}

\arxivnumber{2510.26081}

\title{\boldmath Group-equivariant diffusion models for lattice field theory}


\author[a,b]{Octavio Vega\,\orcidlink{0000-0002-4360-8126},}
\author[c]{Javad Komijani\,\orcidlink{0000-0002-6943-8735},}
\author[a,b]{Aida El-Khadra\,\orcidlink{0000-0001-9105-8213},}
\author[c]{and Marina Marinkovic\,\orcidlink{0000-0002-9883-7866}}
\affiliation[a]{Department of Physics, University of Illinois Urbana-Champaign,\\1110 W Green St, Urbana, IL 61801, U.S.A.}
\affiliation[b]{Illinois Center for Advanced Studies of the Universe,\\1110 W Green St, Urbana, IL 61801, U.S.A.}
\affiliation[c]{Institut für Theoretische Physik, ETH Zürich,\\Wolfgang-Pauli-Str. 27, 8093 Zürich, Switzerland}

\emailAdd{octavio5@illinois.edu}
\emailAdd{jkomijani@gmail.com}
\emailAdd{axk@illinois.com}
\emailAdd{marinama@phys.ethz.ch}

\abstract{Near the critical point, Markov Chain Monte Carlo simulations of lattice quantum field theories become increasingly inefficient due to critical slowing down. In this work, we investigate score-based symmetry-preserving diffusion models as an alternative strategy to sample two-dimensional $\phi^4$ and ${\rm U}(1)$ lattice field theories. We develop score networks that are equivariant to a range of group transformations, including global $\mathbb{Z}_2$ reflections, local ${\rm U}(1)$ rotations, and periodic translations $\mathbb{T}$. The score networks are trained using an augmented training scheme, which improves sample quality in the simulated field theories. We also demonstrate that our symmetry-aware models outperform generic score networks in expressivity and effective sample size.}

\begin{document}
\maketitle
\flushbottom

\section{Introduction}
Lattice Quantum Field Theory (LQFT) provides a non-perturbative avenue for simulating field theories via the path integral formalism~\cite{Feynman:1942us} in Euclidean space-time~\cite{Wilson:1974sk}. In LQFT, Monte Carlo simulation is used to compute expectation values of the observables, which requires averaging over ensembles of field configurations. The efficient sampling of lattice field theories is crucial for non-perturbative studies of fundamental interactions, yet conventional algorithms suffer from poor scalability when approaching the continuum limit~\cite{Fantoni:2021ghj, Albandea:2021kwe}. In this paper, we study Monte Carlo techniques in the framework of diffusion models to address the current limitations of standard sampling approaches.

In LQFT with Euclidean action $S[\phi]$, an expectation value of the observable $\mathcal{O}$ is computed with respect to the path integral measure 
\begin{equation}\label{eq:exp_val_obs}
    \langle \mathcal{O} \rangle = \frac{1}{\mathcal{Z}}\int \mathcal{D} \phi \; \mathcal{O}(\phi) e^{-S[\phi]},
\end{equation}
where $\mathcal{Z}$ is the corresponding Euclidean partition function. The functional integration measure $\mathcal{D}\phi$ represents a sum over all possible configurations of the dynamical degrees of freedom, represented as discretized fields $\phi(x)$ living on lattice sites $x \in \mathbb{Z}^{N_d}$. The expectation value \eqref{eq:exp_val_obs} can be estimated numerically by averaging over batches of samples $\{\phi_{i}\}_{i=1}^N \sim p$ drawn from the probability distribution $p[\phi] = \frac{e^{-S[\phi]}}{\mathcal{Z}}$
\begin{equation}
    \langle \mathcal{O}\rangle \approx \frac{1}{N} \sum_{i=1}^N \mathcal{O}(\phi_i).
\end{equation}
The physical distribution $p[\phi]$ can be arbitrarily complicated, often being defined over high-dimensional, multimodal phase spaces. Additionally, the normalization constant $\mathcal{Z}$ is generally intractable. Sampling therefore requires carefully tailored algorithms and large computational resources. In LQFT, sampling is generally performed with Markov chain Monte Carlo (MCMC) approaches, such as Metropolis-Hastings~\cite{Metropolis:1953am} or Hamiltonian Monte Carlo (HMC)~\cite{Duane:1987de, Gupta:1992np, Kennedy:1993jy}.

Traditional MCMC algorithms suffer from poor scaling with increasing lattice volumes and decreasing lattice spacings. Critical slowing down~\cite{Wolff:1989wq} causes autocorrelation times between samples to grow very large, and topological freezing results in insufficient exploration of the configuration space, manifested though large autocorrelation times of the topological charge~\cite{Woit:1983tq, Polonyi:1983ck, Alles:1996vn}. Machine learning (ML) has been applied to accelerate various stages of the LQFT simulation pipeline, especially ensemble generation~\cite{Boyda:2022nmh, Lawrence:2025rnk}.

Generative artificial intelligence has become increasingly popular in its applications to scientific research and high-performance computing, with their widely developing instances being deep generative models~\cite{Tomczak:2022DeepGenerativeModeling}. There are several classes of generative models with different implementations, capabilities, and limitations, including variational autoencoders~\cite{Kingma2014}, generative adversarial networks (GAN)~\cite{Goodfellow:2014nips}\footnote{See refs.~\cite{Zhou:2018ill, Pawlowski:2018qxs} for applications of GANs to LQFT.}, and normalizing flows~\cite{Rezende:2015icml, Dinh:2017, Papamakarios:2021}.

It has already been established that normalizing flows are capable of probabilistic modeling on non-Euclidean geometries~\cite{Rezende:2020toriSpheres} as well as sampling of various lattice field theories~\cite{Albergo:2019eim, Kanwar:2020xzo, Boyda:2020hsi, Albergo:2021vyo, Albergo:2022qfi, Abbott:2022zhs, Abbott:2022hkm}. Continuous normalizing flows~\cite{Chen2018NeuralODE}, along with stochastic normalizing flows~\cite{Wu:2020SNF}, have also been applied to lattice field theories~\cite{Caselle:2022acb, Gerdes:2022eve, Caselle:2022esc, Caselle:2023mvh, Caselle:2023uel, Caselle:2024ent, Bulgarelli:2024cqc}. Through robust experimentation, progress continues to be made towards understanding the scalability of flow-based sampling towards the continuum for LQFT~\cite{Abbott:2022zsh, Abbott:2023thq, Cranmer:2023MLSampling, Abbott:2024knk, Abbott:2025kvi} and addressing the associated challenges therein~\cite{Komijani:2023fzy, Bulgarelli:2024brv, Komijani:2025yjz}.

Another promising class of deep generative models that have seen tremendous success in an array of contemporary applications are diffusion models. Drawing from the theory of non-equilibrium thermodynamics, diffusion models frame the training and sampling pipelines through the coupling of a forward (diffusion) process and a reverse (denoising) process, respectively. The inference of new samples is formulated as the solution trajectory to a stochastic differential equation that analytically reverses the forward process~\cite{Anderson:1982ReverseDiffusion}, viewed simply as evolution under the heat equation. Diffusion models can be viewed alongside normalizing flows through the lens of the Fokker–Planck equation. The former is described by the diffusive component (heat equation) and the latter by the deterministic component (continuity equation). As opposed to learning bijective maps which deterministically transform data, diffusion models learn a reversible stochastic process whose inverse transports `diffused' configurations back into posterior samples over time.

Diffusion models have been used with remarkable success for computer vision tasks~\cite{Croitoru:2023DiffusionVision} and high-quality image generation with models such as GLIDE~\cite{Nichol:2021GLIDE}, Imagen~\cite{Saharia:2022Imagen}, Stable Diffusion~\cite{Rombach:2022LDM}, and DALL-E 2~\cite{Ramesh:2022esg}. Another emerging application of diffusion models is to drug discovery and \emph{de novo} protein design~\cite{Pang2024DrugGen, Alakhdar:2024DiffusionDrugDesign}, with specific implementations in RFdiffusion~\cite{Watson2023RFdiffusion}, DiffInt~\cite{Sako:2025DiffInt}, and DiffDock~\cite{Corso:2023DiffDock}.

Diffusion models have also recently been studied in the context of lattice field theory~\cite{Wang:2023sry, ranner2024sampling, Habibi:2024fbn}. The authors of refs.~\cite{Wang:2023exq, Aarts:2024agm} provide an interesting discussion on the similarities between score-based generation and stochastic quantization~\cite{Damgaard:1987stochastic} for scalar field theories, upon which they expand for Abelian gauge theory in refs.~\cite{Zhu:2024DiffusionLattice, Zhu:2025pmw}. It is interesting to explore the efficacy of diffusion-based methods for lattice field theories near phase transitions. Furthermore, the pursuit of diffusion-based sampling for lattice QFT invites further studies of phase transitions and critical behavior with potentially greater efficiency.

In this work, we introduce a new framework for sampling in lattice field theory with symmetry-preserving diffusion processes. We apply our method to two primary examples: scalar field theory and Abelian gauge theory in two spacetime dimensions. The main features of our framework are:
\begin{itemize}
    \item \textbf{Generalized diffusion processes:} we present a formalism for stochastic processes which clarifies the role of reversibility in diffusion processes and supports Langevin-based predictor-corrector methods. Our treatment also encompasses both the variance-preserving and variance-expanding pictures of diffusion.
    \item \textbf{Equivariance of score functions:} a key theoretical result derived in this paper is the transformation law of score functions under general symmetry group actions, proved in appendix~\ref{apx:score_func_symmetries}. We discuss this transformation law for abstract groups, including the two main examples studied in this paper: $\mathbb{Z}_2$ and ${\rm U}(1)$. We also provide interpretations of this general result for cases where the underlying data occupies different spaces, such as group manifolds or Lie algebras.
    \item \textbf{Exact symmetry:}  we develop group-equivariant score networks with exact $\mathbb{Z}_2$ and ${\rm U}(1)$ equivariance, and partial translation equivariance. Unlike previous approaches relying on approximate symmetry enforcement via data augmentation, our implementations enforce symmetry by construction.
    \item \textbf{Force-guided score matching:} we introduce a new objective function which augments the traditional score matching training procedure by additionally constraining the learned score to match the theory's force field. This acts as a form of physics-informed guidance, which corrects the score network at the endpoint of the denoising trajectory and yields higher sample quality in our experiments.
\end{itemize}

Our paper is organized as follows: in section~\ref{sec:preliminaries}, we develop the formalism for diffusion processes through stochastic differential equations and discuss our main theoretical result, which is the transformation law for score functions under group symmetries. We further outline the procedure for training score networks and calculating model likelihoods. Next, in section~\ref{sec:0d_toy_example}, we walk through an illustrative example of our pipeline for a $\mathbb{Z}_2$-invariant toy theory. This is followed by section~\ref{sec:phi4_application}, where we present an application of our method to sampling from the $\phi^4$ scalar field theory, accompanied by a discussion on the impacts of enforcing model symmetries. In section~\ref{sec:u1_2d_application}, we apply our procedure to ${\rm U}(1)$ gauge theory in 2D, and demonstrate the validity of the proposed models by computing physical observables. We conclude in section~\ref{sec:conclusion} with a summary and discussion of future directions for this work.
Our appendices expand on the background and formalism for lattice field theory (appendix~\ref{apx:lqft_formalism}), stochastic processes (appendix~\ref{apx:fp_dynamics}), and group symmetries (appendix~\ref{apx:group_symmetries}).

\section{Building diffusion models for LQFT}\label{sec:preliminaries}
Diffusion models, originally inspired by non-equilibrium thermodynamics~\cite{SohlDickstein:2015Deep}, consist of two stochastic processes that evolve data over time. During the forward process, ground-truth data are ``diffused" by random noise until they reach a purely corrupted state where the original signal and its internal correlations are completely erased. After, the reverse process progressively undoes the destruction by ``denoising" samples from a noisy prior to reconstruct realistic samples. In general, both processes are formulated as continuous-time stochastic processes, but can be simulated through discrete-time Markov chains~\cite{Ho2020DDPM}.

In this section, we begin by offering a general formulation of diffusion processes through stochastic differential equations, and explain how they can be reversed. We show how this formalism unifies different frameworks of diffusion models by careful changes of variables. Next, we informally outline our key result that describes how score functions transform equivariantly under group symmetries of the action. Then, we outline how score matching is used to train diffusion models and provide a new addition called \emph{force guidance} which leverages information about the underlying physical theory. We conclude this section by presenting how model likelihoods are computed to allow for bias-corrected calculations.

\subsection{Diffusion as a reversible stochastic process}
A forward diffusion process can be written as a stochastic differential equation (SDE):
\begin{equation}\label{eq:general_fwd_sde}
    d\phi_t = f(\phi_t, t)dt + \sigma dW.
\end{equation}
Here, $\phi_t$ is the state of the system (or field configuration) at diffusion time $t$,
$f(\phi_t, t)$ represents a deterministic drift term, and $dW$ is the increment of a Wiener
process with zero mean and variance $|dt|$. In general, $\sigma$ may also be
time-dependent, but we suppress this for now and work with a constant diffusion coefficient.

By setting the drift term to zero, we can define a diffusion model in the framework of machine learning and train it via the algorithms presented in section~\ref{sec:training}. However, a constant $\sigma$ may require exponentially large time for the diffusion process to lose information about the initial distribution of $\phi_0$. A more efficient approach is to add a non-vanishing drift term $f(\phi_t) = -\gamma\phi_t$ for $\gamma > 0$, which transforms the forward process into
\begin{equation}\label{eq:change_of_vars_VP_SDE}
    d\left(e^{\gamma t}\phi_t\right) = \sigma e^{\gamma t} dW.
\end{equation}
This choice accelerates the loss of information while still enabling the use of the same simple algorithms as those employed in the case of vanishing drift.

\begin{figure}
    \centering
    \includegraphics[width=1.00\linewidth]{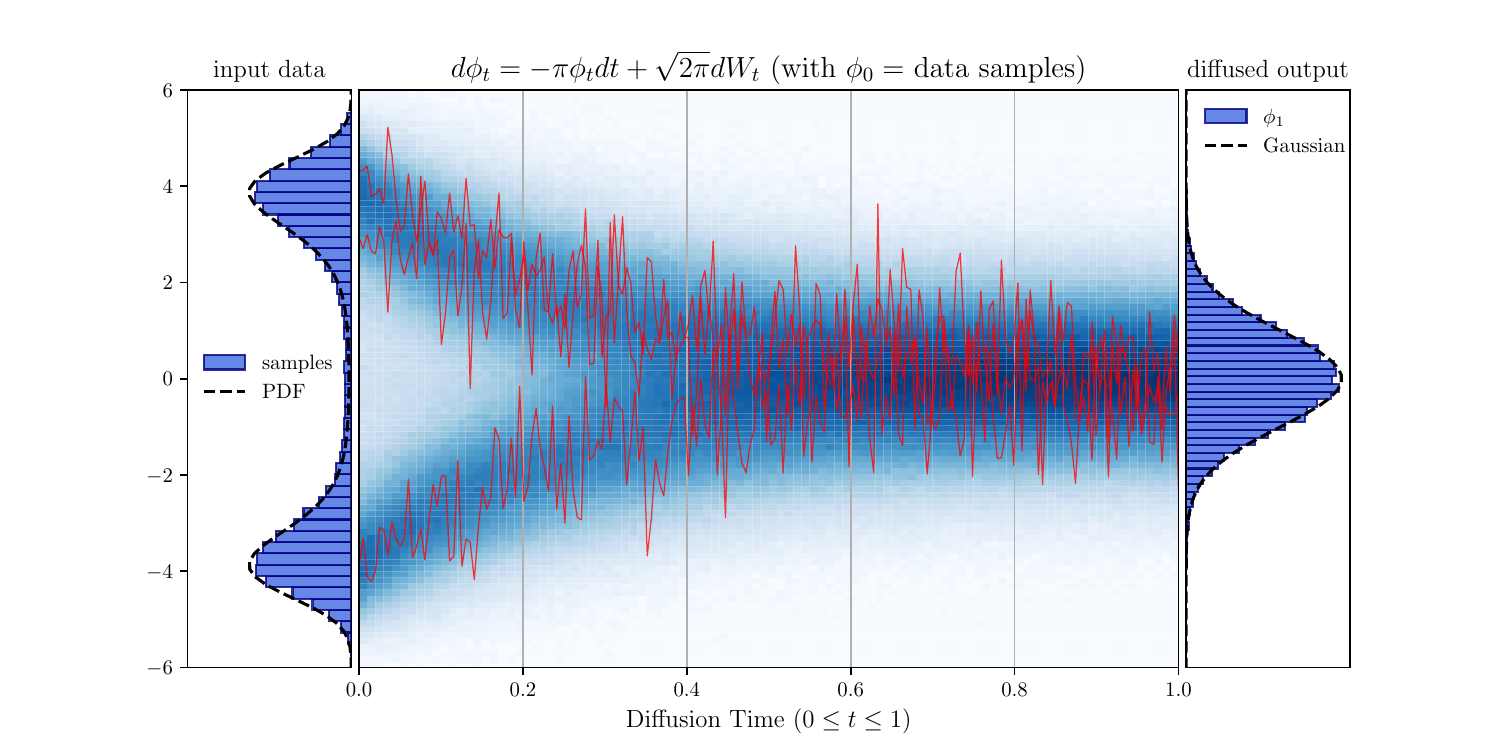}
    \caption{Visualization of the forward diffusion process over time in the variance-preserving picture. The histogram (blue) of data samples is seen to evolve into a Gaussian distribution while the variance is approximately constant in time. The red lines represent three sample trajectories of the data stochastically evolving towards $t=1$.}
    \label{fig:forward_diffusion_visualization}
\end{figure}

The solution to the SDE at time $t > 0$ can be expressed in closed form:
\begin{equation}\label{eq:simplified_vp_example}
  \phi_t = e^{-\gamma t} \phi_0 + \sigma \sqrt{\frac{1 - e^{-2\gamma t}}{2\gamma}} \epsilon_t,
\end{equation}
where $\epsilon_t \sim \mathcal{N}(0, \mathds{1})$. We may set $\sigma = \sqrt{2\gamma}$ so that the variance of $\phi_t$
approaches 1 at large values of $t$, i.e., for $t \gg 1 / \gamma$. This is an example of \emph{variance-preserving} diffusion processes. Alternatively, one can let $\sigma$ depend on time (e.g. as $\sigma_0 e^{\gamma t}$) and suppress the drift term. This is called the \emph{variance-expanding} picture, and it is discussed further below. In figure~\ref{fig:forward_diffusion_visualization}, we show an example diffusion process evolving according to eq.~\eqref{eq:change_of_vars_VP_SDE} with $\gamma \equiv \pi$ and $\sigma \equiv \sqrt{2\pi}$.

The probability density function (PDF) of $\phi_t$ changes as time evolves from the original distribution at time 0 to a Gaussian-like distribution at time 1. Using $p_t(\phi_t; \sigma)$ to denote the PDF of diffused samples at time $t$, the evolution of $p_t$ is governed by the Fokker-Planck equation (FPE)~\cite{Fokker1914, Planck1917}:
\begin{equation}\label{eq:fp_eqn}
    \frac{\partial}{\partial t} p_t(\phi_t; \sigma)
    = \nabla \left(- f(\phi_t, t) + \frac{\sigma^2}{2} \nabla\right) p_t(\phi_t; \sigma)\,.
\end{equation}
For a derivation of the FPE, see appendix~\ref{apx:fp_equation}.

The process in eq.~\eqref{eq:general_fwd_sde} is reversible because the FPEs of the forward and reverse process must be the same. This follows from the fact that a sequential evolution of the random field $\phi$ governed by equation~\eqref{eq:general_fwd_sde} from time $t_1$ to $t_2 > t_1$ and by an evolution in reverse from $t_2$ to $t_1$ does not change the PDF of $\phi$. The reverse operation is commonly referred to as a \emph{denoising} process, for which there exists a family of stochastic processes:
\begin{equation}\label{eq:general_reverse_process}
    d\phi_t = f(\phi_t, t) dt - \frac{\sigma^2 + \tilde\sigma^2}{2}\nabla \log p_t(\phi_t; \sigma)dt + \tilde\sigma dW
\end{equation}
for any real value of $\tilde \sigma$\footnote{Similar to $\sigma$, $\tilde\sigma$ can depend on time, but we drop its time-dependence for simplicity.}. We derive this reverse process in appendix~\ref{apx:reverse_process_derivation}. A common choice is to let $\tilde \sigma = \sigma$. When $\tilde \sigma = 0$, the denoising process in eq.~\eqref{eq:general_reverse_process} simplifies to a completely deterministic trajectory called the \textit{probability flow ODE}, which can be interpreted as a continuous normalizing flow~\cite{Chen2018NeuralODE}.

The more general form of the reverse process in \eqref{eq:general_reverse_process} allows for an intuitive understanding of the sampling process. The reverse SDE can be written as the sum of two terms:
\begin{equation}
    d\phi_t = \underbrace{\left(f(\phi_t, t) - \frac{\sigma^2}{2} \nabla \log p_t(\phi_t; \sigma)\right) dt}_{\rm predictor} 
    + \underbrace{\left(\frac{\tilde\sigma^2}{2} \nabla \log p_t(\phi_t; \sigma)|dt| + \tilde\sigma d W\right)}_{\rm Langevin \; corrector}.
    \label{eq:generic_two-sided_sde}
\end{equation}
The first term (predictor) plays the role of a deterministic drift, ensuring that samples evolve backwards according to the diffusion proccess's marginal probabilities.
The second term (Langevin corrector) has two components that compensate each other; it moves samples towards higher density regions of $p_t$ while also adding random noise in such a way that
it preserves the distribution $p_t$.
The preceding explanation holds under the assumption that the samples $\phi_t$ are indeed drawn from the target distribution $p_t(\phi_t)$.
In practice, when the samples are not correctly drawn, e.g. due to discretization effects, the second term can play the role of a corrector that moves the samples towards the correct distribution in the same spirit that the Langevin process works; hence, it is called a Langevin corrector. From this perspective, our reverse SDE intrinsically describes a predictor-corrector sampling scheme~\cite{Hamming1959}, where the second term represents a single application of a Langevin corrector at time $t$.

The apparent asymmetry between eqs.~\eqref{eq:general_fwd_sde} and \eqref{eq:general_reverse_process} due to the latter's extra drift term can be attributed to the Langevin corrector that does not change under time reversal. Equation~\eqref{eq:generic_two-sided_sde} is written such that it can be applied in both forward and reverse directions; setting $\tilde \sigma = \sigma$ and applying eq.~\eqref{eq:generic_two-sided_sde} in the forward direction yields the original forward process in eq.~\eqref{eq:general_fwd_sde}.

A denoising process aims to recover alternative samples of the original data from the diffused versions by removing the noise that is introduced during the forward diffusion process. A crucial step in applying the denoising process is determining the \emph{score function}, defined as the gradient of $\log p_t(\phi_t)$ with respect to the data:
\begin{equation}\label{eq:score_function}
    \boldsymbol{s}(\phi_t, t) := \nabla \log p_t(\phi_t).
\end{equation}
In score-based generative modeling~\cite{Song2019ScoreBased}, the goal is to learn a neural network approximator for the score function $\hat{\boldsymbol{s}}_\theta(\phi_t, t) \approx \boldsymbol{s}(\phi_t, t)$ over time. We discuss how to train score networks in section~\ref{sec:training}. Given a trained score network, the reverse process in eq.~\eqref{eq:general_reverse_process} can be solved numerically to obtain new samples.

We conclude this section by relating the SDEs of variance-expanding (VE) and variance-preserving (VP) diffusion processes with one another. The VP-SDE is given by
\begin{equation}\label{eq:vp_sde}
    d\phi_t = -\frac{1}{2}\beta(t) \phi_t dt + \sqrt{\beta(t)}dW,
\end{equation}
where $\beta(t)$ is a positive noise schedule. In this framework, the variable $\phi_t$ evolves with fixed variance, as seen in figure~\ref{fig:forward_diffusion_visualization}. This is equivalent to the formulation in eq.~\eqref{eq:simplified_vp_example}, where $\beta(t) \equiv 2\gamma$.
By contrast, the VE-SDE is written as
\begin{equation}\label{eq:ve_sde}
    d\phi_t = \sigma^t dW,
\end{equation}
where $\sigma$ is a constant (analogous to a thermal diffusivity) chosen to control the noise level. This process evolves purely by the addition of noise, causing the variance of the variable $\phi_t$ to rapidly increase over time according to
\begin{equation}
    {\rm Var}[\phi_t] = \frac{\sigma^{2t} - 1}{2 \log\sigma}.
\end{equation}
These two pictures can be unified via a change of variables: multiplying eq.~\eqref{eq:vp_sde} by an appropriate integrating factor gives
\begin{equation}
    d\left[e^{\frac{1}{2}\int_0^t \beta(s) ds} \phi \right] = \sqrt{\frac{d}{dt} \left[e^{\int_0^t \beta(s)ds}\right]} dW,
\end{equation}
which takes the form of a VE-SDE with $\sigma(t)^2 = e^{\int_0^t \beta(s) ds}$. From this point onward, we use the variance-expanding (VE) scheme for our diffusion processes.

\subsection{Symmetries of score functions}\label{sec:symmetries_of_score}
In quantum field theories, the action often exhibits one or more symmetries. We denote with $\mathcal{G}(\phi)$ the group of symmetry transformations that leave the action $S(\phi)$ invariant, $S(\mathcal{G}(\phi))=S(\phi)$. Following eq.~\eqref{eq:score_function}, the score function is derived from the action through differentiation with respect to the field variables:
\begin{equation}
	\boldsymbol{s}(\phi, t)  \triangleq -\nabla S_t^{\rm eff}(\phi),
\end{equation}
which is equivalent to the effective force at time $t$. If $G \in \mathcal{G}$ acts smoothly on $\phi$ through the representation $\rho_G$,
\begin{equation}
	\phi \stackrel{G}{\mapsto} \phi' = \rho_G(\phi),
\end{equation}
then since the action is symmetric under $G$, the probability density must remain unchanged. The score function responds to the group action on the space of fields according to a general transformation law:
\begin{equation}\label{eq:score_fn_transformation_law}
	\boldsymbol{s}(\phi, t) \mapsto \boldsymbol{s}'(\phi', t) = \left( \frac{\partial \rho_G(\phi)}{\partial \phi} \right)^{-\top} \boldsymbol{s}(\phi, t),
\end{equation}
which we state and prove in Theorem~\ref{thm:score_func_transformation}. Although the precise interpretation of eq.~\eqref{eq:score_fn_transformation_law} is representation-dependent, the encompassing result is that the score function generally transforms contravariantly with the Jacobian of the change of variables induced by the group action. In this framework, one can view the score function as \emph{equivariant}. Further details, including derivations, formal discussions, and proofs are provided in appendix~\ref{apx:score_func_symmetries}. Specific examples of this result applied to $\mathbb{Z}_2$ in $\phi^4$ theory and ${\rm U}(1)$ gauge fields can be found in section~\ref{sec:0D_z2_symmetry} and appendix~\ref{apx:u1_symmetry_details}, respectively.

\subsection{Training through force-guided score matching}\label{sec:training}
\emph{Score matching}~\cite{Hyvarinen:2005ScoreMatching} is a technique used to train models to approximate the data distribution $p(\phi)$ by minimizing the difference between the true score $\nabla \log p(\phi)$ and the model score function $\boldsymbol{s}_\theta(\phi)$, where $\theta$ denotes trainable model parameters. As opposed to minimizing the Kullback-Leiber divergence~\cite{KullbackLeibler:1951InfoSuff} between the model and target distributions directly, the score matching objective function compares their gradients and is written as\begin{equation}\label{eq:original_sm_objective}
    J(\theta) = \mathbb{E}_{t \sim \mathcal{U}([0, 1])} \mathbb{E}_{\phi_t \sim p_t(\phi_t)}\left[\lambda(t) \left\|\hat{\boldsymbol{s}}_{\theta}(\phi_t, t) - \nabla \log p_t(\phi_t) \right\|_2^2\right].
\end{equation}
To simplify the score matching procedure, it is often more practical to work with the conditional distribution $p_t(\phi_t | \phi_0)$ instead of directly using the marginal distribution $p_t(\phi_t)$. Under mild regularity assumptions, it can be shown that optimizing \eqref{eq:original_sm_objective} is equivalent to optimizing an alternative objective which only involves the conditional distribution~\cite{Vincent2011ScoreDenoising}, such that the inner expectation is replaced by two nested expectations:
\begin{align}\label{eq:marginal_to_conditional_sm}
    &\mathbb{E}_{\phi_t \sim p_t(\phi_t)}\left[\left\|\hat{\boldsymbol{s}}_{\theta}(\phi_t, t) - \nabla \log p_t(\phi_t) \right\|_2^2\right] \nonumber \\
    &\rightarrow \mathbb{E}_{\phi_0 \sim p(\phi)} \mathbb{E}_{\phi_t \sim p_t(\phi_t | \phi_0)}\left[\left\|\hat{\boldsymbol{s}}_{\theta}(\phi_t, t) - \nabla \log p_t(\phi_t | \phi_0) \right\|_2^2\right] + {\rm const.}
\end{align}
With this objective, training only requires random sampling from the dataset ($\phi_0 \sim p_0(\phi)$) and the conditional distribution ($\phi_t \sim p_t(\phi_t |\phi_0)$). The rewritten objective function is expressed as a weighted integral over the temporal interval $0 \leq t \leq 1$:
\begin{equation}\label{eq:conditional_score_matching}
    J(\theta) = \int_0^1 \lambda(t) \, \mathbb{E}_{\phi_0 \sim p(\phi)} \mathbb{E}_{\phi_t \sim p_t(\phi_t | \phi_0)}\left[\left\|\hat{\boldsymbol{s}}_{\theta}(\phi_t, t) - \nabla \log p_t(\phi_t | \phi_0) \right\|_2^2\right] dt.
\end{equation}
The uniform measure with respect to which the outermost expectation over time is taken is easily sampled; however, computing the integral in eq.~\eqref{eq:conditional_score_matching} is often memory intensive. Instead, one may sample a time $t \sim \mathcal{U}([0, 1])$ and simply compute the loss at this single time step. This estimate converges to the full loss and provides a faster way of training. 

Note that the integration constant appearing in eq.~\eqref{eq:marginal_to_conditional_sm} is not necessarily finite. To regulate the loss function, an appropriate weight function is needed, for which we choose a simple weight function $\lambda(t) = \sigma(t)^2$ that vanishes as $t \to 0$. Consequently, training based on this loss function alone gains limited direct insight from the data when $t \ll 1$, relying instead on extrapolation to cover this region. To shift from extrapolation to interpolation, we can incorporate prior knowledge of the score function at $t = 0$ by augmenting the loss function with an additional term as follows:
\begin{equation}\label{eq:force_regularization}
    \mathcal{J}(\theta) = J(\theta) + c_0 \mathbb{E}_{\phi_0 \sim p_0(\phi_0)} \left[\|\hat{\boldsymbol{s}}_\theta(\phi_0, 0) + \nabla S[\phi_0]\|_2^2 \right],
\end{equation}
where $c_0 \geq 0$ is a regularization coefficient. Although this augmentation may not be useful for many applications of diffusion models, it can readily be applied in lattice field theory,
where $\partial_{\phi_0} \log p_0(\phi_0)$ (the force) is analytically known. We summarize our training procedure in Algorithm~\ref{alg:train_0D_dm}.
\begin{algorithm}
\caption{Force-Guided Score Matching}\label{alg:train_0D_dm}
\begin{algorithmic}[1]
    \State \textbf{Input:} Training data $\Phi_0$, variance function $\sigma_t$, neural net $\boldsymbol{s}_\theta$, action $S$
    \State \textbf{Initialize:} Parameters of model $\theta$
    \For{each epoch}
        \State ${\rm Loss} \leftarrow 0$
        \For{each minibatch $\{\phi_0\} \subseteq \Phi_0$}
            \State Sample $t \sim \mathcal{U}([0, 1])$ for each config
            \State Sample $\epsilon \sim \mathcal{N}(0, \mathds{1})$ for each config
            \State $\phi_t = \phi_0 + \sigma_t \epsilon$ \qquad\qquad\qquad\qquad\qquad\qquad\qquad \qquad\qquad  // \texttt{diffuse data (VE)}
            \State $\mathcal{J}(\theta) = \left\|\sigma_t\boldsymbol{s}_\theta(\phi_t, t) + \epsilon\right\|_2^2 + c_0 \left\|\boldsymbol{s}_\theta(\phi_0, 0) + S'(\phi_0)\right\|_2^2$ \;\;\;\; //
            \texttt{estimate loss}
            \State ${\rm Loss} \leftarrow \mathcal{J}(\theta) + {\rm Loss}$
            \EndFor
        \State Compute gradient: $\nabla_\theta {\rm Loss}$
        \State Update parameters $\theta$
    \EndFor 
    \State \Return $\hat{\boldsymbol{s}}_\theta$
\end{algorithmic}
\end{algorithm}

\subsection{Likelihood computations and exactness}\label{sec:likelihood_computation}
Generative models for lattice field theory must be exact, meaning that generated samples must be asymptotically distributed according to the target distribution. The discrepancy between the learned and target probabilities in the ML models can be corrected given that the model probability of samples $q(\phi)$ can be computed.

In principle, the model log-likelihood of generated samples can be computed simultaneously as samples are generated. To illustrate, consider evolving backwards in time by a small step size $h$:
\begin{equation}
    \begin{pmatrix}
        \phi_t \\ \log q_t(\phi_t) 
    \end{pmatrix} 
    \rightarrow 
    \begin{pmatrix}
        \phi_{t - h} \\ \log q_{t-h}(\phi_{t - h}) 
    \end{pmatrix} 
\end{equation}
Then, as discussed previously, the sample evolves first by a deterministic drift, and is then corrected by an extra term involving noise:
\begin{equation}
    \phi_t \xrightarrow[\rm predictor]{-h \left[f(\phi_t, t) - \frac{\sigma^2}{2}\boldsymbol{s}(\phi_t, t)\right]}
    \phi_{t - h}'
    \xrightarrow[\rm corrector]{+h \frac{\tilde \sigma^2}{2}\boldsymbol{s}(\phi_t, t) + \tilde\sigma \sqrt{h}\eta_t} \phi_{t - h}
\end{equation}
The corrector step shifts the sample to another sample also distributed according to the same marginal density via
\begin{equation}
    \phi_{t-h}' \rightarrow \phi_{t-h} = \phi_{t - h}' \underbrace{+ h \frac{\tilde\sigma^2}{2}\boldsymbol{s}(\phi_t, t) + \sqrt{h}\tilde\sigma \eta_t}_{\Delta \phi}.
\end{equation}
The drift term evolves $\phi_t$ according to the probability flow ODE, along which $\log q_t$ evolves accordingly. However, the corrector step induces an extra discretization error:
\begin{equation}
    \log q_t(\phi_t + \Delta \phi) = \log q_t(\phi_t) + \boldsymbol{s}(\phi_t, t) \Delta\phi + O((\Delta \phi)^2).
\end{equation}
In our framework, the step size can be effectively controlled by the parameter $\tilde\sigma$ which is decoupled from the step size $h$ and the forward noise $\sigma$. In practice, it is common to see only the probability flow ODE ($\tilde\sigma = 0$) used for generation when computing model likelihoods to avoid compounding discretization effects.

Assuming a variable $\phi_t$ evolves according to eq.~\eqref{eq:general_fwd_sde}, then eq.~\eqref{eq:fp_eqn} can be used to solve an ODE for $\log q_t$ along deterministic trajectories in reverse time, which results in
\begin{equation}\label{eq:prob_flow_solution}
    \log q_0(\phi_0) = \log q_1(\phi_1) + \int_0^1 \nabla \cdot \boldsymbol{f}_\theta(\phi_t, t) \, dt,
\end{equation}
where
\begin{equation}
    \boldsymbol{f}_\theta(\phi_t, t) := f(\phi_t, t) - \frac{\sigma^2}{2} \hat{\boldsymbol{s}}_\theta(\phi_t, t).
\end{equation}
We provide the argument for this solution in appendix~\ref{apx:fp_ode_solution}. In general, computing the divergence term inside the integrand of \eqref{eq:prob_flow_solution} is intractable; however, one can make use of the \emph{Skilling-Hutchinson trace estimator}~\cite{Skilling1989, Hutchinson1989}
\begin{equation}
    \nabla \cdot \boldsymbol{f}_\theta \triangleq {\rm Tr}\left[J_{\boldsymbol{f}_\theta}\right] = \mathbb{E}_{\epsilon \sim \mathcal{N}(0, \mathds{1})}\left[\epsilon^\top J_{\boldsymbol{f}_\theta} \epsilon\right].
\end{equation}
to obtain an unbiased estimate for the divergence. 

Equipped with the ability to compute model probabilities, one gains access to two avenues for ensuring exactness in later calculations, which we briefly detail here:
\begin{itemize}
    \item \textbf{Reweighting}: A change of probability measure allows one to reweight~\cite{Noe2019} their observables directly inside the expectation:
    \begin{equation}
        \langle \mathcal{O}(\phi) \rangle_p = \left\langle\mathcal{O}(\phi) \frac{p(\phi)}{q(\phi)}\right\rangle_q .
    \end{equation}
    \item \textbf{Resampling}: By resampling~\cite{Albergo:2019eim} a generated ensemble with an accept / reject step, one can construct an asymptotically exact Markov chain with an acceptance probability given by
    \begin{equation}
        A(\phi' | \phi) = \min\left(1, \frac{p(\phi') q(\phi)}{q(\phi') p(\phi)}\right).
    \end{equation}
\end{itemize}
The method one chooses to perform exact calculations depends on how efficient drawing samples from a trained model is compared to computing observables.

\section{Toy example in 0D}\label{sec:0d_toy_example}
In this section, we illustrate the pipeline of equivariant score-based generative modeling with our symmetric diffusion models using a simple example involving a univariate probability density.

\subsection{$\mathbb{Z}_2$ symmetry}\label{sec:0D_z2_symmetry}
We define a toy theory on a zero-dimensional lattice through the action
\begin{equation}\label{eq:0D_action}
    S(\phi) = \frac{1}{2}m^2\phi^2 + \frac{\lambda}{4!}\phi^4,
\end{equation}
where the variable $\phi$ corresponds to one degree of freedom occupying a single lattice site. Notably, this theory exhibits two distinct phases that depend on the action parameters:
\begin{itemize}
    \item \textit{symmetric phase}: $\frac{\lambda}{m^2} > 0$ (single well potential)
    \item \textit{broken phase}: $\frac{\lambda}{m^2} < 0$ (double well potential)
\end{itemize}
which correspond to the topological phases of the model. In this example, we explore both phases by choosing the following action parameters: $m^2 = \pm 1.0$, $\lambda = 0.4$.

This toy model is clearly $\mathbb{Z}_2$-invariant, since the action in eq.~\eqref{eq:0D_action} enjoys a discrete global symmetry under $\phi \mapsto -\phi$. This means that the samples, modulo random statistical fluctuations, should be distributed symmetrically about $\phi~=~0$. This is in accordance with the general result, derived in section~\ref{sec:symmetries_of_score}; in the case where $\mathcal{G} = \mathbb{Z}_2$, then clearly $\rho(\phi) = -\phi$, so eq.~\eqref{eq:score_fn_transformation_law} reduces to
\begin{equation}
	\boldsymbol{s}(\phi, t) \mapsto -\boldsymbol{s}(\phi, t).
\end{equation}
In other words, the score function for a zero-dimensional scalar field theory with global $\mathbb{Z}_2$ invariance is, in turn, $\mathbb{Z}_2$-equivariant. The $\mathbb{Z}_2$ symmetry of eq.~\eqref{eq:0D_action} is explicit when computing the force:
\begin{equation}
    -\frac{\partial S(\phi)}{\partial \phi} = -m^2 \phi - \frac{\lambda}{3!}\phi^3,
\end{equation}
which is clearly odd in $\phi$. This inspires us to construct a score network which is also $\mathbb{Z}_2$-equivariant such that the diffusion model need not learn this discrete symmetry of the theory, but is endowed with it instead. 

For our score network, we use a simple multilayer perceptron (MLP)~\cite{almeida2020multilayer} with a custom-built antisymmetric sigmoidal activation function, defined in terms of the usual sigmoid function $\sigma(z)$ as
\begin{equation}
    {\rm OddSigmoid}(z) := \frac{1}{2} - \frac{e^{-z}}{1 + e^{-z}} = \sigma(z) - \frac{1}{2}.
\end{equation}
The manifest symmetry in the MLP score networks is displayed in figure~\ref{fig:0D_learned-score_funcs}.
\begin{figure}
    \centering
    \includegraphics[width=1.00\linewidth]{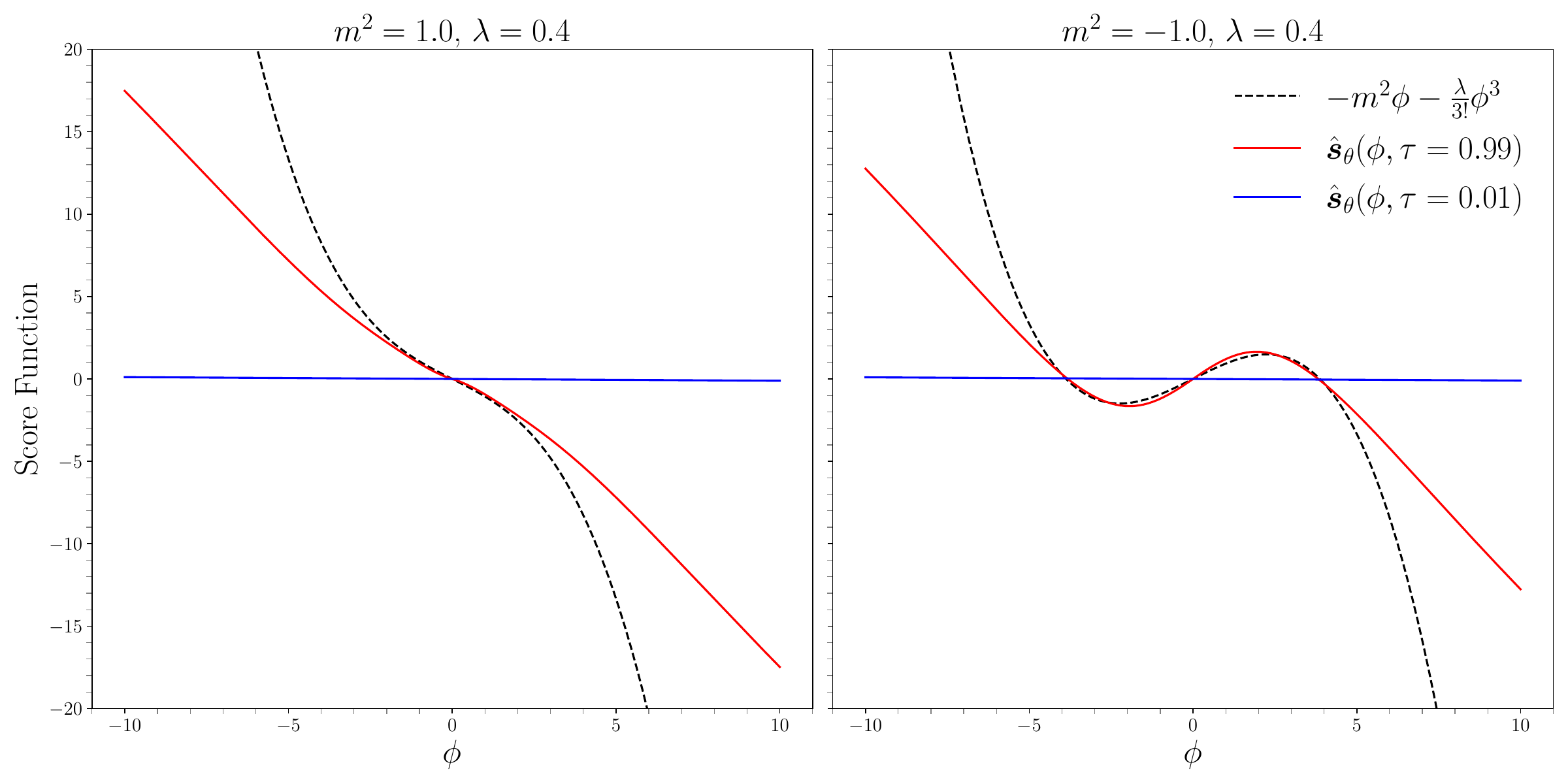}
    \caption{Learned scores as functions of the field variable $\phi$ in the (left) symmetric and (right) broken phases. The blue and red solid lines correspond to the score functions near the start ($\tau  = 0.01$) and end ($\tau  = 0.99$) of the reverse process, respectively. The dotted black line represents the true force, $-\nabla S(\phi)$.}   
    \label{fig:0D_learned-score_funcs}
\end{figure}
Since the score function must learn to denoise data $\phi_t$ conditioned on time $t$, the score network must also be time-dependent. The raw input time would be a scalar $t \in (0, 1)$, but to properly condition a neural network, a more expressive representation is needed. In score-based diffusion models, the general practice is to lift $t$ into a higher-dimensional space through a process called \emph{time embedding}. Different methods exist for embedding sequential information, for instance, sinusoidal positional encodings~\cite{Vaswani2017Attention} or Gaussian Fourier projections~\cite{Tancik2020Fourier}.

Diffusion time is a continuous variable, of which diffusion models must learn complex time-dependent score functions. In light of this, we choose to encode time via Gaussian Fourier projections, where random frequencies $f_i$ are chosen from a Gaussian distribution and fixed according to a predefined embedding dimension $d$. Then the encoding is built as
\begin{equation}
    t \mapsto \left[\sin(2\pi f_1 t), \cos(2\pi f_1 t), ..., \sin(2\pi f_d t), \cos(2\pi f_d t)\right].
\end{equation}
Projecting times into a random, higher-dimensional Fourier basis improves the universal function approximation capacity of the model to learn arbitrary smooth functions of time. This method of time embedding in diffusion models has been seen to perform better than the geometric progression of frequencies in sinusoidal positional encodings~\cite{Karras2022Elucidating}.

\subsection{Training and inference}
For training, we use HMC to prepare ensembles of 16,384 independently generated samples in both the single and double well phases. We train the diffusion models in this example with the Adam optimizer~\cite{Kingma2015Adam} for 250 epochs with a learning rate of 0.0001 and batch sizes of 256. The training proceeds as outlined in Algorithm~\ref{alg:train_0D_dm}.

For inference, we generate new samples by solving the reverse ODE using a numerical Euler integrator with a step size of 0.002 (500 time steps). In order to illustrate sample quality, we also employ a similar Euler-Maruyama integrator to solve the reverse SDE, noting that likelihood calculation is best performed in tandem with deterministic generation. We display samples from our trained models compared to the HMC data on which they were trained for both phases in figure~\ref{fig:0D_sample_histograms}.
\begin{figure}
    \centering
    \includegraphics[width=0.95\linewidth]{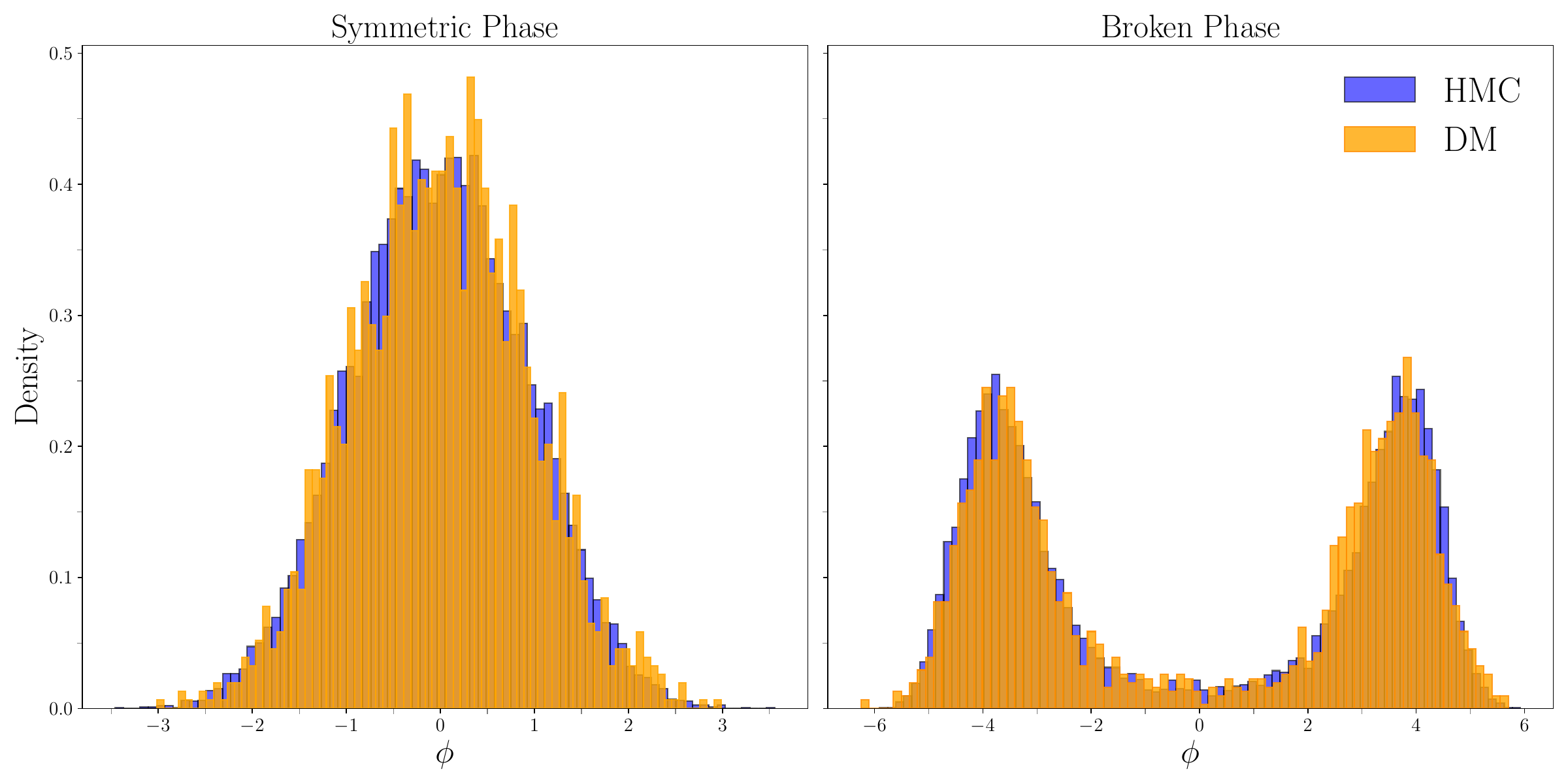}
    \caption{Normalized histograms with 80 bins over 16,384 samples of training data (HMC) and 2,048 samples of diffusion model (DM) generated data for (left) the symmetric and (right) broken phases. The coupling strength was set to $\lambda = 0.4$ for both phases, and the squared mass is $m^2 = 1.0$ and $m^2 = -1.0$ for the two phases, respectively.}
    \label{fig:0D_sample_histograms}
\end{figure}

As shown in figure~\ref{fig:0D_learned-score_funcs}, the models learn the score functions very well within regions of $\phi$ centered around the local minima of the action, since these correspond to areas of high probability. The regions beyond the middle zone where the score network seems to fail to learn the force accurately correspond to areas of much lower likelihood, meaning much fewer samples of data exist there. Because they are manifestly odd in $\phi$ and must therefore vanish at $\phi=0$, the models approximate the score functions well near $\phi = 0$. Non-odd score networks are also capable of good performance, though we prefer to work with $\mathbb{Z}_2$-equivariant networks in order to exactly replicate the theory's symmetry and because we observe higher sample quality with symmetric networks, as discussed below in table~\ref{tab:model_quality_metrics}.

At each intermediate time step $\tau \in [0, 1]$ along the reverse diffusion process, the learned score function defines, up to additive constants, an effective action $S_\tau^{\rm eff}(\phi)$ corresponding to the learned density at that time. For this simple one-dimensional example, the effective action can be easily computed and visualized by cumulatively integrating the learned score functions with respect to $\phi$ over fixed intervals $[\phi_{\rm min}, \phi] \subset \mathbb{R}$:
\begin{equation}
	S^{\rm eff}_\tau(\phi) = -\int_{\phi_{\rm min}}^{\phi} \hat{\boldsymbol{s}}_\theta(\tilde{\phi}, \tau) \, d \tilde{\phi}.
\end{equation}
To illustrate, we plot the trajectory of learned effective actions for the broken phase in figure~\ref{fig:0D_broken_effective_action}.
\begin{figure}
    \centering
    \includegraphics[width=0.75\linewidth]{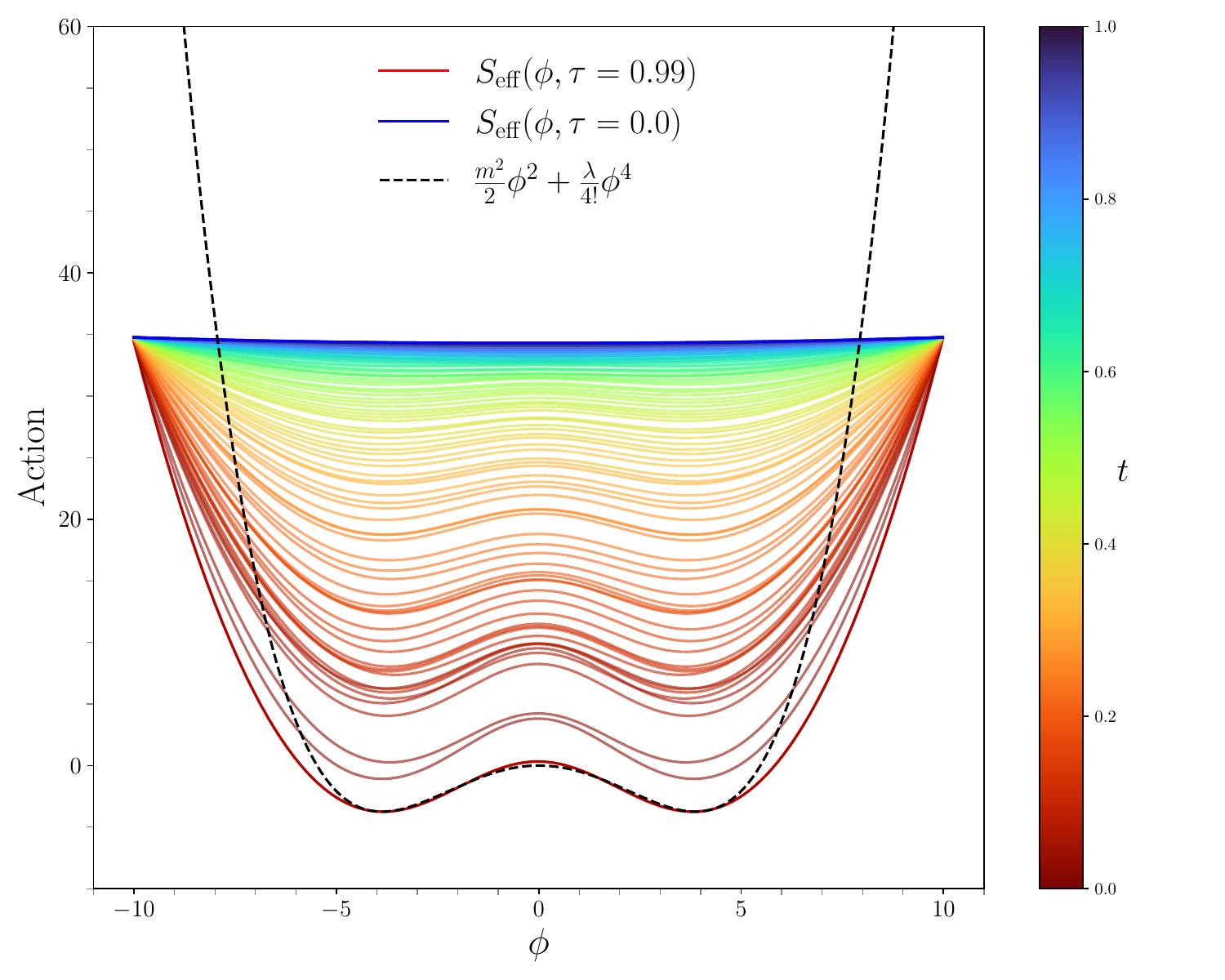}
    \caption{Evolution of the effective action in the broken phase as a function of the field $\phi$ over reverse diffusion time. The colorbar is parameterized by (forward) diffusion time $t \equiv 1 -\tau$.}
    \label{fig:0D_broken_effective_action}
\end{figure}
Through the effective action $S_\tau^{\rm eff}(\phi) := -\log q_{\tau}(\phi)$, the learned score function also implicitly defines (up to normalization) an effective probability density of the generated samples at any point $\tau$ in the reverse process. More visualization of the learned effective actions along with their approximate densities over diffusion time $\tau$ and the field $\phi$ are seen in figure~\ref{fig:0D_action_and_probability_surfaces}.
\begin{figure}
    \centering
    \includegraphics[width=1.00\linewidth]{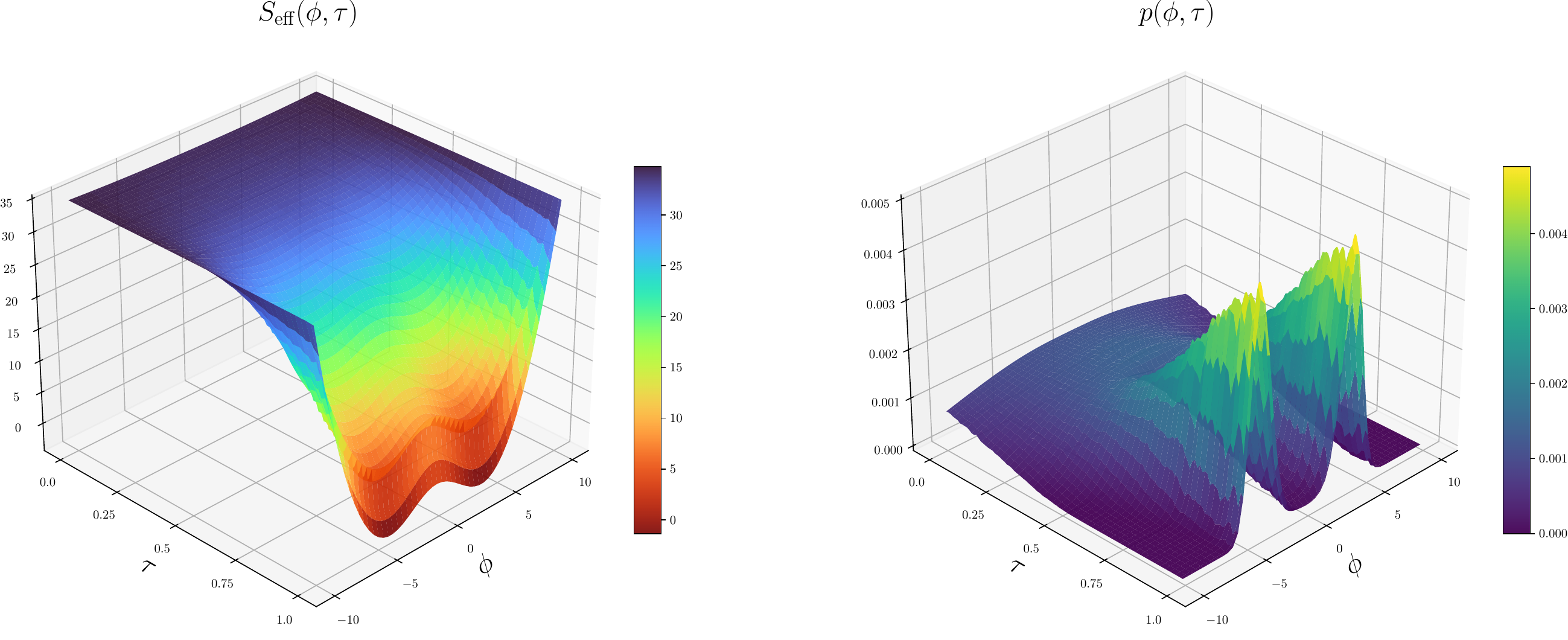}
    \caption{Evolution of (left) the effective action $S_{\rm eff}(\phi, t)$ and (right) the resulting density $p(\phi, t) \propto e^{-S_{\rm eff}(\phi, t)}$ over reverse diffusion time $\tau$ as a function of the field $\phi$ in the broken phase.}
    \label{fig:0D_action_and_probability_surfaces}
\end{figure}

\subsection{Effective sample sizes}\label{sec:0D_ess}
At this point, given the ease of interpretation of this toy example, we review the concept of  effective sample size (ESS)~\cite{Kong1994Sequential}. Given a target distribution $p$ and a model distribution $q$, the ESS is defined as
\begin{equation}\label{eq:ess_from_probs}
    {\rm ESS} := \frac{\left(\frac{1}{N}\sum_{i=1}^N p[\phi_i] / q[\phi_i]\right)^2}{\frac{1}{N}\sum_{i=1}^N \left(p[\phi_i] / q[\phi_i]\right)^2},
\end{equation}
where $N$ is the number of samples in an ensemble of configurations $\{\phi_i\}_{i=1}^N$. 
Given that the target distribution is defined by a physical action $S[\phi] =: -\log p[\phi]$ and the model likelihood induces an effective action $S_{\rm eff}[\phi] := -\log q[\phi]$, eq.~\eqref{eq:ess_from_probs} can also be written as
\begin{equation}
    {\rm ESS} = \frac{\left(\frac{1}{N} \sum_{i=1}^N w(\phi_i)\right)^2}{\frac{1}{N} \sum_{i=1}^N w(\phi_i)^2},
\end{equation}
where the $w(\phi_i)$ represent weighting factors given by
\begin{equation}
    w(\phi_i) = \exp\left(-S[\phi_i] + S_{\rm eff}[\phi_i]\right).
\end{equation}
These weighting factors allow us to reframe the discussion in the context of importance sampling, where samples from $q$ are used to approximate expectations under $p$:
\begin{equation}
    \mathbb{E}_p\left[\mathcal{O}(\phi)\right] = \int \mathcal{D}\phi \; \mathcal{O}(\phi) p[\phi] = \int \mathcal{D}\phi \; \mathcal{O}(\phi) w(\phi) q[\phi].
\end{equation}
When the weights $w_i \equiv w(\phi_i)$ are such that only a few samples dominate, then this estimate becomes poor, resulting in an effective number of samples much smaller than the total number $N$. In this way, we view the ESS as the fraction of generated configurations that are equivalent to statistically independent samples. Furthermore, the ESS takes values in $[0, 1]$, where ${\rm ESS} = 1$ corresponds to a perfect model where all generated configurations are weighted equally and are truly i.i.d. from the target. We obtain effective sample sizes of 99.4\% and 94.6\% in the symmetric and broken phases, respectively.

\section{Application to scalar $\phi^4$ theory in 2D}\label{sec:phi4_application}
The scalar $\phi^4$ theory is often cited as the simplest interacting quantum field theory that still captures many of the key nontrivial features of more complex theories. For example, many systems near critical points in statistical field theory behave like the $\phi^4$ theory~\cite{Zinn-Justin:1989rgp}, making it a prototype for phase transitions and studying universality. The $\phi^4$ theory is also an important playground for studying non-perturbative dynamics on the lattice, including spontaneous symmetry breaking and triviality in 4D~\cite{Montvay:1994cy}. It often serves as a testbed for algorithm development, as many lattice algorithms are first tested on $\phi^4$ before being applied to more complex theories~\cite{Gattringer:2010zz}.

The lattice action for the scalar $\phi^4$ theory can be written as
\begin{equation}\label{eq:phi4_lattice_action}
    S[\phi] = \sum_x \left[-\kappa \sum_\mu \phi(x) \phi(x + \hat{\mu}) + \left(1 - 2\lambda\right)\phi(x)^2 + \lambda \phi(x)^4 \right],
\end{equation}
which, in addition to the local potential terms from the zero-dimensional example, now includes extended nearest-neighbor interaction terms that allow the field to propagate through spacetime. Notably, \eqref{eq:phi4_lattice_action} exhibits a global $\mathbb{Z}_2$ invariance as well, meaning that its force (and therefore the score function), following the discussion in section~\ref{sec:0D_z2_symmetry}, should also be $\mathbb{Z}_2$-equivariant.

\subsection{Network architecture}\label{sec:phi4_network_architecture}
We borrow a well-known neural network architecture from computer vision applications, known as the U-Net, which was first used for semantic segmentation in biomedical imaging~\cite{Ronneberger2015UNet}. The U-Net consists of two sequential components: a contracting (encoding) path and an expanding (decoding) path. The encoder consists of several convolutional layers which gradually downsample the input image while increasing the number of channels; afterwards, the decoder upsamples the image again using transposed convolutions while simultaneously decreasing the number of channels back down to 1. Additionally, residual (skip) connections in the network concatenate data in the decoder with the intermediate data from the encoding path at the same spatial resolution. 
\begin{figure}
    \centering
    \includegraphics[width=1.00\linewidth]{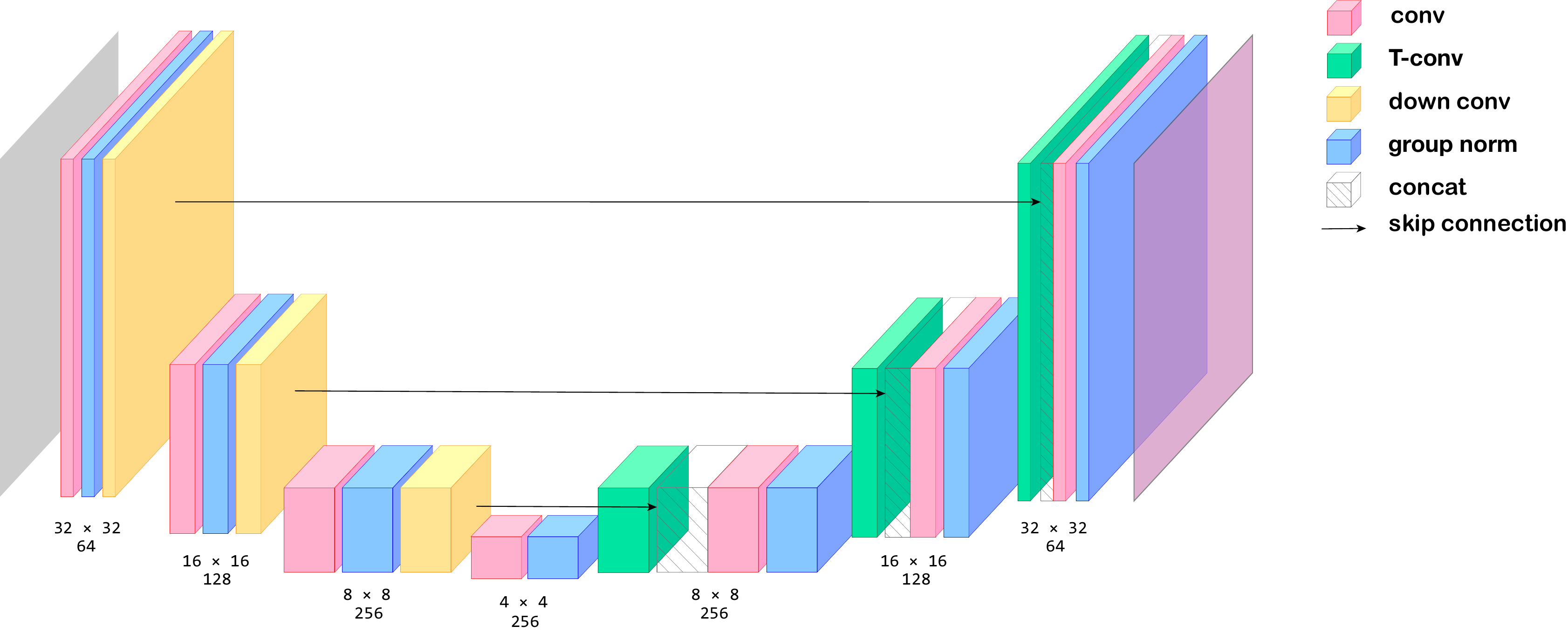}
    \caption{Schematic of the U-Net architecture used in our experiments to parameterize the score network, in the example of a $32 \times 32$ input. At each stage, convolutional blocks downsample the resolution of the configurations while increasing the intermediate channels, after which transposed convolutions perform the reverse operation. Beneath each convolutional block, the top pair of numbers represents the shape of the data at that resolution and the bottom number is the corresponding number of channels. The input and output are assumed to have a single channel, and are shaped $(\cdot, 1, L, L)$.}
    \label{fig:unet_diagram}
\end{figure}

In figure~\ref{fig:unet_diagram}, we display an example diagram of our U-Net architecture that showcases how the lattice geometry evolves as it moves through the network. In each convolutional block, a convolutional layer is applied which preserves the data's dimensions while increasing the number of channels. This is followed by the application of a group normalization layer~\cite{Wu2018GroupNorm} and another convolutional layer which downsamples the lattice by a factor of two in each direction. In the expanding path, we use transposed convolutions to restore the spatial resolution of the data.

In two spatial dimensions, scalar field configurations can be thought of as images; here, the `pixels' correspond to the field values at each of the corresponding lattice sites $\phi(x) \in (-\infty, \infty)$, and the image width/height are the spatial lattice extents $L \times L$. Similarly, instead of RGB channels, we have a single channel. In figure~\ref{fig:phi4_reverse_lattices}, we illustrate an example of a $32 \times 32$ $\phi^4$ configuration evolving along the reverse diffusion process in similar fashion to a 2D image. Unlike images, however, field configurations and the lattices on which they reside are endowed with certain physical symmetries, which we show can be built in to score-based diffusion models via the design of the score network itself. In section~\ref{sec:impact_of_model_symmetry}, we detail the beneficial impact that respecting these symmetries has on our model quality.

Our current implementation of a $\mathbb{Z}_2$-equivariant network is to take a generic (non-equivariant) U-Net and anti-symmetrize it:
\begin{equation}\label{eq:antisymmetric_score_func}
    \overline{\boldsymbol{s}}_\theta(\phi, t) = \frac{1}{2}\left(\hat{\boldsymbol{s}}_\theta(\phi, t) - \hat{\boldsymbol{s}}_\theta(-\phi, t)\right)
\end{equation}
such that the resulting score function $\overline{\boldsymbol{s}}_\theta$ is odd by construction. This enables us to use a conventional, unrestricted implementation of a U-Net (with circularly padded convolutional layers and SiLU activation functions) to parameterize a still $\mathbb{Z}_2$-equivariant score. In practice, this requires twice as many score network evaluations during inference. Such a $\mathbb{Z}_2$-equivariant U-Net can alternatively be constructed by using layers which respect this symmetry and odd activation functions, as we did in section~\ref{sec:0D_z2_symmetry}. 

The theory also respects an important spacetime symmetry. Namely, in LQFT simulations, the lattices obey periodic boundary conditions as a means of replicating an infinite spacetime continuum while avoiding edge artifacts. To take periodic boundary conditions into account, we use circular padding in the convolutional layers of our score network. Our numerical experiments demonstrate that incorporating this explicitly, as opposed to using traditional padding, enhances model quality.

\subsection{Comparison with HMC}\label{sec:phi4_comparison_with_hmc}
For all of our numerical experiments, we focus on the broken phase, where $\frac{\lambda}{m^2}~<~0$. The specific parameters we choose are $m^2 = -2.68$, $\lambda = 0.5$, $\kappa = 0.67$. They correspond to the broken phase, but not overly far from the critical point.
Our training dataset consists of 16,384 configurations which were generated in parallel using 200 HMC trajectories for thermalization (with hot initialization). Each HMC trajectory has unit length, and consists of 20 leapfrog steps. We train for 250 epochs with a batch size of 128 at a learning rate of 0.0001.
\begin{figure}
    \hspace{-1.0cm}
    \includegraphics[width=1.00\linewidth]{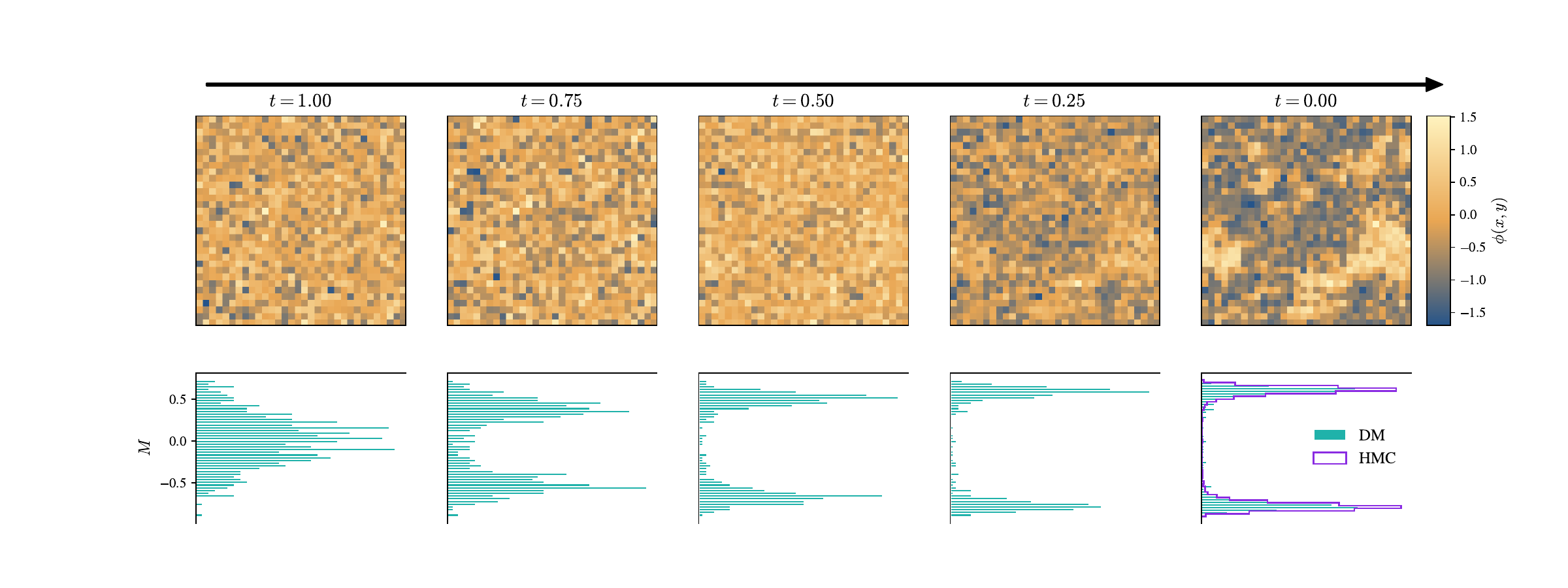}
    \caption{The top row pictures a sample $\phi^4$ configuration deep in the broken phase on a $32 \times 32$ lattice evolving over reverse diffusion time from left to right. The sample progresses from pure noise to a realistic configuration with local correlations. The bottom row displays the evolution of the magnetization ($M$) histograms over reverse diffusion time, and the last panel compares the resulting diffusion model (DM) density with the density computed on an HMC-generated dataset. This $32 \times 32$ example is meant as illustration to provide a higher resolution image, but our numerical experiments in this section focus on $8 \times 8$ and $16 \times 16$ lattices.}
    \label{fig:phi4_reverse_lattices}
\end{figure}

For inference, we use a fourth order Runge-Kutta (RK4) integrator~\cite{Runge1895, Kutta1901} to solve the probability flow ODE with a step size of 0.002 (500 time steps) to generate new samples $\phi$. We simultaneously solve the probability flow ODE for the model log-likelihood $\log q[\phi]$ using the Skilling-Hutchinson estimator with 10 random Rademacher~\cite{Hitczenko1994} vectors for the divergence estimate. Together, the backwards process amounts to solving the following system of ODEs:
\begin{equation}
    \frac{d}{dt}
    \begin{pmatrix}
        \phi_t \\ \log q_t[\phi_t]
    \end{pmatrix}
    = 
    -\begin{pmatrix}
        \frac{1}{2}\sigma^{2t} \hat{\boldsymbol{s}}_\theta (\phi_t, t) \\ \nabla \cdot  [\sigma^{2t}\hat{\boldsymbol{s}}_\theta(\phi_t, t)]
    \end{pmatrix}.
\end{equation}

We quantitatively validate our models by computing physically relevant quantities and comparing them either to previously-existing estimates computed on Monte Carlo data, their exactly known values (when available), or both. We focus on the following physical observables relevant to the $\phi^4$ theory for validation (where $V = L^2$ denotes the lattice volume):
\begin{itemize}
    \item Magnetization (mean-field value):
    \begin{equation}\label{eq:magnetization}
        M = \frac{1}{V}\sum_x \phi(x)
    \end{equation}
    \item Magnetic susceptibility:
    \begin{equation}\label{eq:2pt_suscept}
        \chi = V \left( \langle M^2 \rangle - \langle M \rangle^2 \right)
    \end{equation}
    \item Modified magnetic susceptibility:
    \begin{equation}\label{eq:mod_2pt_suscept}
        \chi^\prime = V \left( \langle M^2 \rangle - \langle |M| \rangle^2 \right)
    \end{equation}
    \item Binder cumulant \cite{Binder1981}:
    \begin{equation}\label{eq:binder_cumulant}
        U_L = 1 - \frac{1}{3}\frac{\langle M^4 \rangle}{\langle M^2 \rangle^2}\,.
    \end{equation}
\end{itemize}
We introduced $\chi^\prime$ so that in the broken-phase region one can effectively calculate the magnetic susceptibility for each sector separately~\cite{Kotze2008, DelDebbio:2021vyx}.

Our testing dataset is comprised of $1024$ chains of configurations generated independently using HMC. 
Each HMC chain contains 100 configurations generated in sequence:
the first 100 HMC iterations are for thermalization, followed by an additional 100 HMC iterations that we save. Each HMC trajectory contains 20 leapfrog steps of step size 0.05.
For estimating autocorrelations we use all of the testing dataset,
while for computation of observables in table~\ref{tab:phi4_obs} we only use the last configuration of each set, i.e., 1024 configurations.

We show values of the observables \eqref{eq:magnetization}-\eqref{eq:binder_cumulant} computed across the diffusion and HMC ensembles in Table~\ref{tab:phi4_obs}. We observe that, after reweighting, observables computed on the diffusion model data agree with both training and testing data.
\begin{table}
    \centering
    \begin{tabular}{l l |
                    S[table-format=1.4(4)]
                    S[table-format=1.4(4)]
                    S[table-format=1.4(4)]}
        \specialrule{.1em}{0pt}{0pt}
        \rowcolor{gray!10}
        \textbf{Data Set} & $N_{\rm conf}$ & {$\langle |M| \rangle$} & {$\chi'$} & {$U_L$} \\
        \specialrule{.05em}{0pt}{0pt}
        DM (Raw)            & 1024  & 0.5781(72) & 3.38(15) & 0.535(24) \\
        DM (Reweighted)     & 1024  & 0.6289(72) & 3.35(20) & 0.559(26) \\
        HMC (Testing)       & 1024  & 0.6143(74) & 3.66(16) & 0.547(22) \\
        HMC (Training)      & 16384 & 0.6188(19) & 3.63(04) & 0.548(05) \\
        \specialrule{.1em}{0pt}{0pt}
    \end{tabular}
    \caption{Mean values and standard errors for the observables defined in Eqs.~\eqref{eq:magnetization}, \eqref{eq:mod_2pt_suscept}, and \eqref{eq:binder_cumulant} computed on ensembles of $8\times8$ $\phi^4$ configurations. The comparison is between generated diffusion model (DM) data before and after reweighting with the HMC training and testing sets. The size of each dataset is labeled by the number of configurations $N_{\rm conf}$.}
    \label{tab:phi4_obs}
\end{table}

The autocorrelation function for an observable $\mathcal{O}$ is defined as
\begin{equation}
    \rho_\mathcal{O}(t) := \langle \mathcal{O}_{t_0} \mathcal{O}_{t_0 + t}\rangle - \langle \mathcal{O}_{t_0}\rangle \langle\mathcal{O}_{t_0 + t}\rangle,
\end{equation}
where $t_0$ is the source time and $t$ is the temporal separation along the Markov chain. We compare the normalized autocorrelation functions,
$\overline{\rho}_{\mathcal{O}}(t) := \rho_\mathcal{O}(t) / \rho_\mathcal{O}(0)$,
for a resampled diffusion-generated ensemble with HMC testing data in figure~\ref{fig:normalized_autocorr_M}.
\begin{figure}
    \centering
    \includegraphics[width=1.0\linewidth]{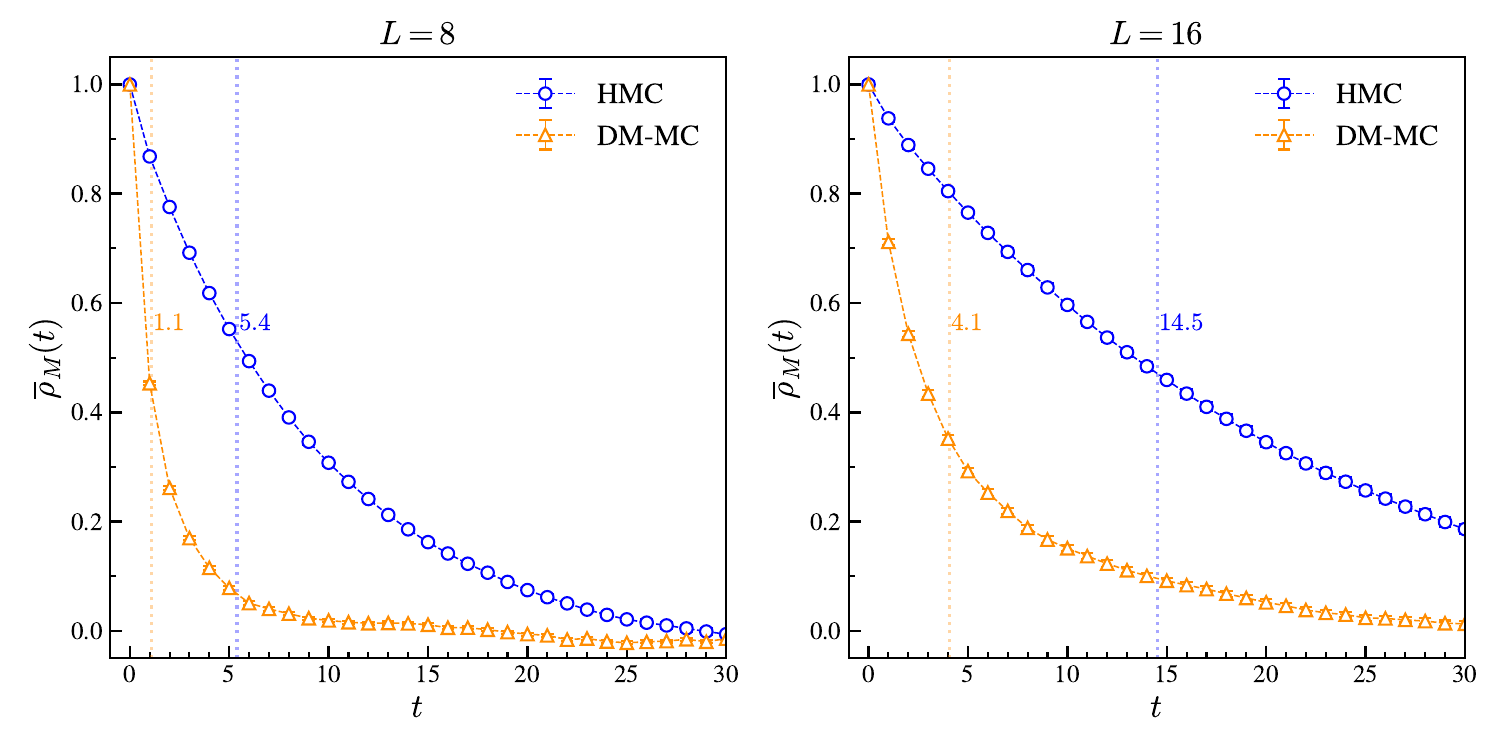}
    \caption{Normalized autocorrelation functions for the magnetization $M$, computed on $8 \times 8$ (left) and $16 \times 16$ (right) lattices from an HMC dataset and a resampled diffusion model dataset (DM-MC). The dotted lines represent the integrated autocorrelation times $\tau_{\rm int}$ for each of the two ensembles.}
    \label{fig:normalized_autocorr_M}
\end{figure}
We also compute the \emph{integrated autocorrelation times} as in \cite{Wolff:2003sm} for $M$ on both datasets. The integrated autocorrelation time of a Markov chain of measurements of an observable $\mathcal{O}$ is defined as
\begin{equation}\label{eq:tau_int}
    \tau_{\rm int}(\mathcal{O}) := \frac{1}{2} + \lim\limits_{t_{\rm max} \to \infty}\sum_{t=1}^{t_{\rm max}} \overline{\rho}_{\mathcal{O}}(t).
\end{equation}

In figure~\ref{fig:logp_vs_logq_hists}, the HMC dataset is generated using 100 independent thermalization trajectories and an additional 20,000 HMC steps, where every 20 samples are saved for computing the autocorrelation. Each HMC tracjectory consisted of 10 leapfrog steps at a step size of 0.1. Similarly, the diffusion dataset was generated by sampling from the trained model
using the Euler method with 200 ODE steps. 
We generate 1000 proposed samples and perform an accept/reject step on them to make a Markov chain with 1000 samples. We repeat the described procedure 100 times and illustrate the average in the figure.
Using $t_{\rm max} = 100$, we find $\tau_{\rm int}^{\rm HMC}(M) = 5.4$ and $\tau_{\rm int}^\text{DM-MC}(M) = 1.1$ at $L=8$; at $L=16$ we obtain $\tau_{\rm int}^{\rm HMC}(M) = 14.5$ and $\tau_{\rm int}^\text{DM-MC}(M) = 4.1$ evaluated on the average of the autocorrelation times over independent chains.

We note that while 200 leapfrog steps per configuration (10 leapfrog steps per trajectory, and 20 trajectories between saved configurations) are used for the HMC dataset, the diffusion model data in the comparison seen in figure~\ref{fig:normalized_autocorr_M} was generated with 200 ODE steps. This comparison indicates that the diffusion model reduces autocorrelations for the lattice sizes and observables studied. We emphasize that these comparisons are valid only for this lattice size, and we leave a more detailed study of the volume scaling to future work.

An additional computational expense for the diffusion model arises from evaluating the U-Net.
The average cost of each U-Net evaluation on batched data (1000 samples) increases by about 45\% from $0.009\:s$ to $0.013\:s$ on a single NVIDIA GH200 GPU for $L=8$ and $L=16$ lattices, respectively. However, the cost of the Hutchinson estimates, which dominates the total cost, is almost the same for both lattice sizes: $0.302\:s$, on average. Therefore, one step of diffusion generation takes nearly the same amount of time for both lattice sizes. For 200 ODE steps, this gives a total cost for the full reverse process of about $200 \times (0.01 + 0.30) = 62\:s$.
For HMC, the cost of 10 leapfrog steps is almost the same for both lattice sizes, taking about $0.0033\:s$ on average. As the lattice size increases, the cost of a leapfrog step should also increase, but this is obfuscated here by the overhead cost of function calls. Overall, the cost of 20 trajectories with 10 leapfrog steps each is roughly $0.066\:s$. 
For the diffusion model, it takes $62\:s$ to generate 1000 biased configurations simultaneously, corresponding effectively to 1000 $\times$ ESS independent samples. By contrast, the HMC sampling must be done sequentially; generating a single chain of 1000 HMC samples with 20 trajectories in between saved configurations requires about $66\:s$ on average.
    
Additionally, we can choose more precise integrators to decrease the number of inference steps when sampling from a diffusion model. In contrast, HMC requires the introduction of auxiliary momenta and one must always use a symplectic integrator, such as the leapfrog method.

\subsection{Influence of symmetry on model expressivity}\label{sec:impact_of_model_symmetry}
A simple measure of model quality is to compare the true action $S[\phi]$ with the learned effective action $S_{\rm eff}[\phi]$, each evaluated on the same batch of generated configurations. If the model is trained well, one should find $S_{\rm eff}[\phi] \approx S[\phi] + c$, where $c$ is an irrelevant normalization constant. Quantitatively, we can judge this alignment through the coefficient of determination $R^2$ between the two actions. In figure~\ref{fig:logp_vs_logq_hists}, we show the learned action in strong agreement with the true action on an $8 \times 8$ lattice.
\begin{figure}
    \centering
    \includegraphics[width=1.00\linewidth]{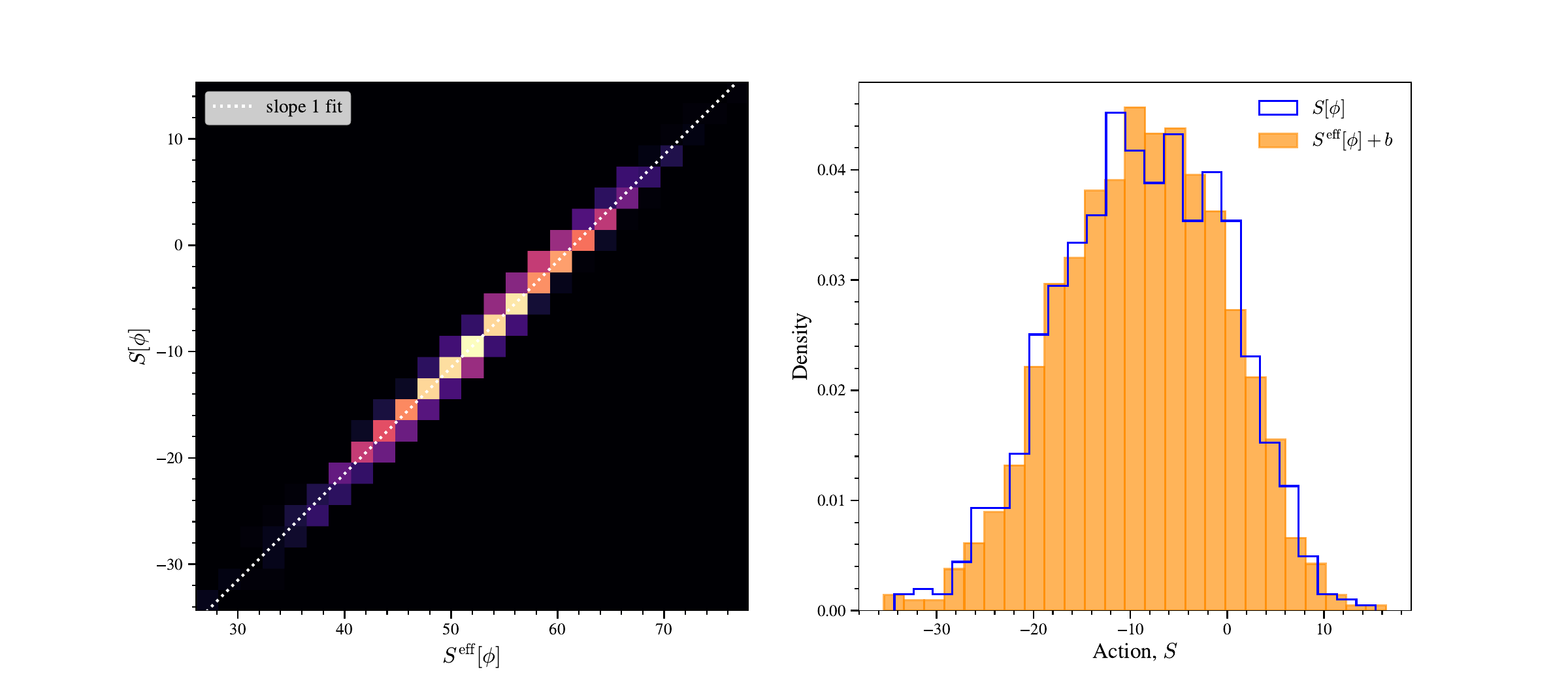}
    \caption{(Left) 2D histograms over 25 bins showing the distributions of the true action $S[\phi]$ against the effective action $S_{\rm eff}[\phi]$ computed on an ensemble of 1024 $\phi^4$ configurations sampled from a trained diffusion model near the critical point. The white dotted line represents a slope 1 fit with affine shift computed as the average difference between the actions: $b = \overline{S[\phi] - S^{\rm eff}[\phi]}$. The correlation between the true and effective action is observed to be high; the corresponding $R^2$ coefficient is 0.9951. (Right) Histograms of the true action and effective action shifted by the fit constant $b$.}
    \label{fig:logp_vs_logq_hists}
\end{figure}

Another measure of model quality is the effective sample size (ESS). As defined and discussed in section~\ref{sec:0D_ess}, the ESS contains information about how importance weights are distributed across a generated ensemble. Moreover, the ESS is also correlated with the acceptance rate of a Metropolis accept/reject chain applied to the ensemble. The probability of accepting the next `proposal' sample $\phi_i$ given the previous sample $\phi_{i-1}$ is given by
\begin{equation}
    P_{\rm accept}(\phi_{i} | \phi_{i-1}) = \min\left(1, \frac{p[\phi_i] / p[\phi_{i-1}]}{q[\phi_i] / q[\phi_{i-1}]} \right) = \min \left(1, \frac{w_i}{w_{i-1}}\right),    
\end{equation}
where the $w_i$ are sample importance weights given by
\begin{equation}
    w_i = \exp\left( -S[\phi_i] + S_{\rm eff}[\phi_i]\right).
\end{equation}
A low ESS indicates that the weights are heavily unbalanced, so the ratio of successive weights will often be small, leading to a low MCMC acceptance rate. Conversely, for a relatively higher ESS, the weights will be more stable and evenly spread, thus we expect higher acceptance rates when resampling our ensemble. 

Notably, both the ESS and acceptance rate are highly sensitive to the reweighting factors $w_i$, whereas the $R^2$ coefficient is less sensitive to the overall spread in the $w_i$. Moreover, these quantities display strikingly different scaling; we consistently obtain $R^2~>~0.95$ even when ESS and acceptance rates fall well below 10\%. For this reason, we maintain that though the $R^2$ provides some heuristic insight into model performance, the ESS and acceptance rate are more stringent benchmarks for the quality of generated samples.

As discussed in section~\ref{sec:phi4_network_architecture}, the $\phi^4$ theory is $\mathbb{Z}_2$-invariant, motivating the use of a $\mathbb{Z}_2$-equivariant score network. Additionally, the periodicity of the lattice inspired us to use circular convolutional layers to respect the spacetime symmetry of the theory as well. Next we investigate the impact that these choices have on model performance and the quality of the samples produced therefrom.

\begin{table}
    \centering
    \setlength{\tabcolsep}{4.5pt}
    \renewcommand{\arraystretch}{1.2}
    \footnotesize
    \arrayrulecolor{black}
    \begin{tabular}{c|cccc|cccc}
        \specialrule{.12em}{0pt}{0pt}
        \rowcolor{gray!20}
        \textbf{Sym.} & \multicolumn{4}{c|}{\textbf{No Regularization}} & \multicolumn{4}{c}{\textbf{Force Reg.} ($c_0 = 0.1$)} \\
        \specialrule{.05em}{0pt}{0pt}
        & NFE $(\downarrow)$ & ESS $(\uparrow)$ & AR $(\uparrow)$ & KL $(\downarrow)$
        & NFE $(\downarrow)$ & ESS $(\uparrow)$ & AR $(\uparrow)$ & KL $(\downarrow)$ \\
        \specialrule{.12em}{0pt}{0pt}

        % Euler integrator (500 steps)
        \rowcolor{gray!10}
        \multicolumn{9}{c}{\textit{Euler Integrator} ($N_{\rm step} = 500$)} \\
        \specialrule{.05em}{0pt}{0pt}
        $\{e\}$                          &  500 & 0.103(15) & 0.284(12) & -59.58(05) &  500 & 0.282(23) & 0.434(16) & -60.22(16)  \\
        $\mathbb{Z}_2$                   & 1000 & 0.201(22) & 0.382(13) & -60.07(04) & 1000 & 0.558(26) & 0.597(12) & -60.57(02) \\
        $\mathbb{T}$                     &  500 & 0.280(21) & 0.462(14) & -60.33(03) &  500 & 0.545(20) & 0.602(10) & -60.59(01) \\
        $\mathbb{Z}_2 \times \mathbb{T}$ & 1000 & 0.527(44) & 0.593(22) & -60.60(03) & 1000 & 0.672(21) & 0.671(13) & -60.70(01) \\
        \specialrule{.12em}{0pt}{0pt}

        % RK4 integrator (125 steps)
        \rowcolor{gray!10}
        \multicolumn{9}{c}{\textit{Runge--Kutta (RK4) Integrator} ($N_{\rm step} = 125$)} \\
        \specialrule{.05em}{0pt}{0pt}
        $\{e\}$                           & 500 & 0.131(18) & 0.304(15) & -59.85(10)  & 500 & 0.264(20) & 0.430(14) & -60.77(14) \\
        $\mathbb{Z}_2$                   & 1000 & 0.209(18) & 0.396(16) & -60.53(11) & 1000 & 0.531(24) & 0.573(12) & -61.16(16) \\
        $\mathbb{T}$                      & 500 & 0.289(28) & 0.450(19) & -60.63(08)  & 500 & 0.522(33) & 0.602(11) & -61.03(10) \\
        $\mathbb{Z}_2 \times \mathbb{T}$ & 1000 & 0.579(34) & 0.613(20) & -61.09(08) & 1000 & 0.597(28) & 0.653(08) & -61.24(18) \\
        \bottomrule
    \end{tabular}

    \caption{Comparison of ESS, KL divergences, and MCMC acceptance rates (AR) for U-Net symmetry variants evaluated on 1024 diffusion-generated $8 \times 8$ $\phi^4$ configurations in the broken phase (repeated 10 times for error estimation). Rows are grouped by ODE integrator (\textit{Euler} and \textit{RK4}), while columns compare standard training (No Reg.) versus force-regularized training ($c_0 = 0.1$). The NFE subcolumn denotes the total number of score function evaluations during the ODE solve. The symbols in the Sym. column are shorthand notation for the symmetry group respected by the corresponding U-Net variant. The arrows ($\uparrow \downarrow$) indicate whether the corresponding quantity is preferred to take higher or lower values.}
    \label{tab:model_quality_metrics}
\end{table}

There are four possible symmetry combinations: no symmetry ($\{e\}$), $\mathbb{Z}_2$ symmetry only, periodic translational symmetry ($\mathbb{T}$), as well as combined $\mathbb{Z}_2$ and $\mathbb{T}$ symmetry. We compare the impacts of including each of these symmetry combinations on the ESS, KL divergence, and acceptance rate for our trained models over ten repeated experiments performed with the same set of parameters during training and sampling. Consistently, we observe that including each symmetry yields better performance across all numerical metrics, with the highest level of performance being obtained when both symmetries are included. We summarize these results in the upper left quadrant of table~\ref{tab:model_quality_metrics} for the specific case of an Euler integrator with $500$ inference steps for $L=8$. We display the same set of results for the $16 \times 16$ lattices in table~\ref{tab:model_quality_metrics_L=16}.

\begin{table}
    \centering
    \setlength{\tabcolsep}{4.5pt}
    \renewcommand{\arraystretch}{1.2}
    \footnotesize
    \arrayrulecolor{black}
    \begin{tabular}{c|cccc}
        \specialrule{.12em}{0pt}{0pt}
        \rowcolor{gray!10}
        \textbf{Sym.} & NFE $(\downarrow)$ & ESS $(\uparrow)$ & AR $(\uparrow)$ & KL $(\downarrow)$ \\
        \specialrule{.05em}{0pt}{0pt}
        $\{e\}$                          &  500 & 0.025(05) & 0.100(14) & -238.83(15) \\
        $\mathbb{Z}_2$                   & 1000 & 0.092(14) & 0.232(17) & -240.16(08) \\
        $\mathbb{T}$                     &  500 & 0.271(21) & 0.370(17) & -240.79(08) \\
        $\mathbb{Z}_2 \times \mathbb{T}$ & 1000 & 0.310(33) & 0.402(17) & -241.04(06) \\
        \specialrule{.12em}{0pt}{0pt}
    \end{tabular}
    \caption{Same as the upper left quadrant of table~\ref{tab:model_quality_metrics}, but for $L=16$.}
    \label{tab:model_quality_metrics_L=16}
\end{table}

Furthermore, the force guidance scheme detailed in eq.~\eqref{eq:force_regularization} provides a family of possible training regimes, parameterized continuously by the regularization coefficient $c_0$. We observe that, when tuned properly, including the force guidance term dramatically improves model quality across all choices of symmetry when compared with models trained using traditional score matching ($c_0 \equiv 0$). In table~\ref{tab:model_quality_metrics}, we compare standard training and force-guided training with $c_0 = 0.1$ on $8 \times 8$ lattices. 

Finally, we emphasize that generating these samples with an RK4 integrator was one choice out of many possible methods to solve the probability flow ODE. In practice, diffusion models enjoy a plethora of sampling methods for improving predictions and reducing discretization errors. For example, higher-order integrators may be used to solve the differential equations with better precision. Other sampling techniques involving predictor-corrector methods or adjustmment techniques, such as the Metropolis-adjusted Langevin algorithm (MALA)~\cite{RobertsTweedie1996Langevin}, are also available. 

For this study, we limited our scope to Euler and RK4 integrators for generating samples. The upper and lower panels of table~\ref{tab:model_quality_metrics} display a comparison between model quality when using each of these ODE solvers. These effects must simultaneously take into account the computational cost that each integration technique incurs. We measure this cost with the number of function evaluations (NFE), which counts how many times the score network must be evaluated during each integration step and multiplies this count by the total number of ODE solver iterations. The baseline is 1 evaluation for Euler and 4 evaluations for RK4 (per integration step), but these are doubled for the $\mathbb{Z}_2$-equivariant network due to the antisymmetrization. On average, we observe comparable performance in the three model quality metrics when using the RK4 integrator with four times fewer steps than the Euler integrator.

\section{Application to ${\rm U}(1)$ gauge theory in 2D}\label{sec:u1_2d_application}
We use the angular representation of ${\rm U}(1)$ gauge theory in two spacetime dimensions and define the ${\rm U}(1)$ lattice action as 
\begin{equation}\label{eq:u1_lattice_action}
    S[\theta] = -\beta \sum_x \sum_{\mu < \nu} \cos \phi_{\mu\nu}(x),
\end{equation}
where the angular plaquettes depend on the gauge link phases $\theta_\mu(x)$ via
\begin{equation}
    \phi_{\mu\nu}(x) = \theta_\mu(x) + \theta_\nu(x + \hat{\mu}) - \theta_\mu(x + \hat{\nu}) - \theta_\nu(x),
\end{equation}
and $\phi_{\mu\nu}$ form gauge-invariant objects on the lattice.
For further details on the formulation of ${\rm U}(1)$ gauge theory on the lattice, we direct the reader to appendix~\ref{apx:formulation_u1_gauge_theory}.

\subsection{Gauge-equivariant score network}
The force for the ${\rm U}(1)$ gauge theory on the lattice is given by
\begin{equation}\label{eq:u1_score_function}
    \boldsymbol{s}(\theta_\mu(x)) = \beta \sum_\nu \left[\sin\phi_{\mu\nu}(x) - \sin \phi_{\mu\nu}(x - \hat{\nu})\right],
\end{equation}
which corresponds to the true score function with respect to the real-valued $\theta_\mu(x)$. For a derivation of equation~\eqref{eq:u1_score_function}, see appendix~\ref{apx:u1_score_func}. The force in equation~\eqref{eq:u1_score_function} is gauge-invariant, as it only depends on plaquette phases $\phi_{\mu\nu}$. We therefore aim to construct a score network that respects gauge symmetry.

To build an explicitly gauge-invariant score network, we make our network a function of gauge-invariant quantities. On their own, the gauge links angles $\theta_0(x), \theta_1(x)$ are gauge-variant, but can be composed into plaquette phases $\phi_{01}(x)$, thereby giving a set of gauge-invariant data. Then, any transformation that the score network applies will remain agnostic to gauge transformations on the links. In contrast to other approaches like data augmentation where one applies random gauge transformations to the input training data, our approach exactly enforces ${\rm U}(1)$ gauge symmetry, as was done in ref.~\cite{ranner2024sampling}.

To parameterize our score network, we again use a U-Net with circularly-padded convolutional layers to respect the lattice's periodic boundaries, as described in section~\ref{sec:phi4_network_architecture}. The gauge-invariant U-Net is constructed such that it takes in data with a single channel (the plaquette phases), and returns data with two channels ($N_d = 2$ spacetime directions) to represent a score function with respect to gauge links.

We train on 65,536 ${\rm U}(1)$ configurations generated with 1000 HMC trajectories (following 200 thermalization steps) of step size 0.1 and 10 leapfrog steps per trajectory. All models were trained for 500 epochs with learning rate $5 \times 10^{-4}$ and batch size 512. The testing dataset to which we compare consists of 8192 configurations generated via the same method.

\begin{figure}
    \centering
    \includegraphics[width=1.00\linewidth]{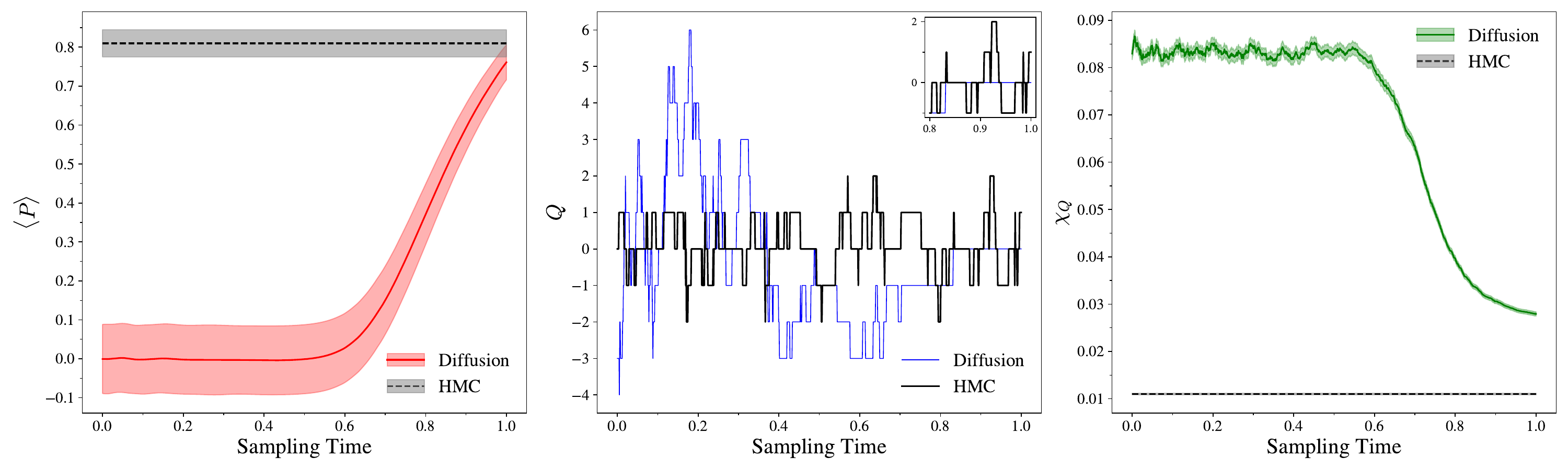}
    \caption{Progression over sampling time of (left) the average plaquette $\langle P \rangle$ with 1$\sigma$ bands and (right) the topological susceptibility $\chi_Q$ with standard error on $8 \times 8$ lattices at $\beta=3.0$. As expected, the values of $\langle P \rangle$ and $\chi_Q$ from the diffusion model are seen to gradually approach those of the HMC test ensemble. The remaining gap at the end of sampling can be corrected after by reweighting (see table~\ref{tab:observables_comparison}). (Middle) topological charge $Q$ of a randomly chosen configuration each from the diffusion-generated and HMC ensembles.}
    \label{fig:u1_reverse_sampling_dynamics}
\end{figure}

After training, we visualize the inference process by tracking certain quantities over reverse diffusion time. An important quantity for tracking the exploration of phase space is the \emph{topological charge}~\cite{Luscher1982Topology}. In two dimensions, the U(1) theory topological charge is defined by
\begin{equation}\label{eq:topo_charge}
    Q := \frac{1}{2\pi} \sum_x \phi_{01}(x) \in \mathbb{Z}.
\end{equation}
In figure~\ref{fig:u1_reverse_sampling_dynamics} we plot the evolution of the average plaquette, topological charge, and topological susceptibility during the reverse process. While we see variations of the topological charge over a broad range during the diffusion-based sampling, the effect of autocorrelations of the topological charge on the observables of interest is yet to be explored. The inset figure displays how, near the end of sampling time, the topological charge of the diffusion-generated configuration ceases to fluctuate. Previous work on diffusion models for image generation has demonstrated that high-level attributes\footnote{We direct the reader to the discussion provided in Appendix D of Ref.~\cite{Ho2020DDPM}.} are decided earlier in the reverse trajectory, and samples beyond this point conditioned on the same latents share these global features~\cite{Ho2020DDPM}. In analogy, the topological charge is also a global feature of U(1) configurations, possibly explaining why it stops fluctuating earlier on.

As was done for the magnetization in the $\phi^4$ theory in section~\ref{sec:phi4_comparison_with_hmc}, we compute the autocorrelation functions for the topological charge $Q$ to compare the diffusion-generated data with an HMC ensemble at couplings $\beta \in \{3.0, 4.0, 5.0\}$. The HMC dataset is generated using 100 independent thermalization trajectories and an additional 10,000 HMC steps, where every 10 samples are saved for computing the autocorrelation. Each HMC tracjectory consisted of 10 leapfrog steps at a step size of 0.1. We plot the normalized autocorrelation functions for $Q$ on two different lattice sizes in fig.~\ref{fig:normalized_autocorr_Q}.
\begin{figure}
    \centering
    \includegraphics[width=1.0\linewidth]{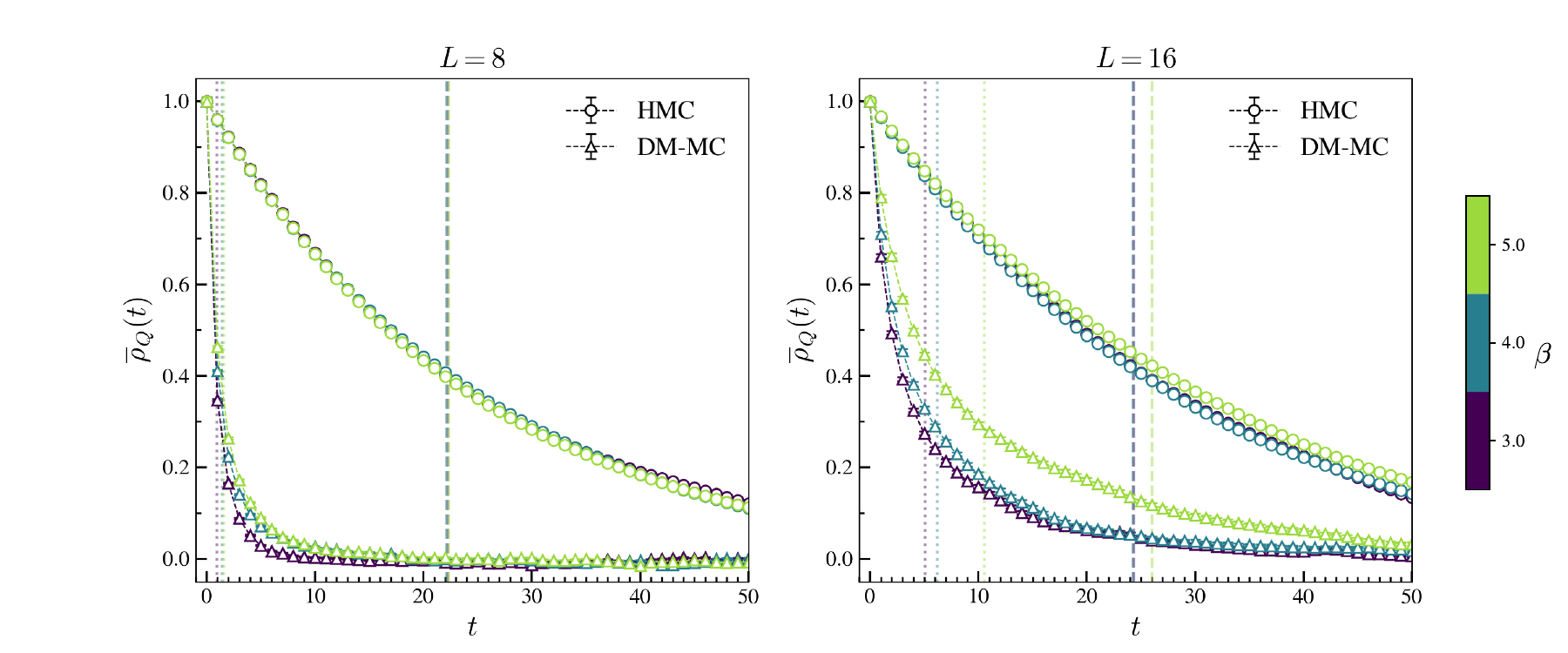}
    \caption{Normalized autocorrelation functions for the topological charge $Q$ at $\beta \in \{3.0, 4.0, 5.0\}$ on (left) $8 \times 8$ and (right) $16 \times 16$ lattices computed for both diffusion model and HMC ensembles. The integerated autocorrelation times for each ensemble are displayed as vertical lines (dotted for diffusion, dashed for HMC).}
    \label{fig:normalized_autocorr_Q}
\end{figure}
For the integrated autocorrelation times (in order of increasing $\beta$), at $L=8$ we find $\tau_{\rm int}^{\rm DM}(Q) = \{0.9, 1.4, 1.5\}$ and $\tau_{\rm int}^{\rm HMC}(Q) = \{22, 22, 22\}$; at $L=16$ we find $\tau_{\rm int}^{\rm DM}(Q) = \{5.1, 6.2, 10\}$ and $\tau_{\rm int}^{\rm HMC}(Q) = \{24, 24, 26\}$.

\subsection{Additional observables}
In the U(1) theory, we compute expectation values of rectangular Wilson loops. In figure~\ref{fig:wilson_loop_comparisons}, we compare Wilson loops between the HMC and diffusion-generated data with the exact results on a $8 \times 8$ lattice at $\beta = 3.0$ and see good agreement. Unbiased results are obtained using the weighting factors $w_i = p_i/q_i$ to compute reweighted means and errors of the observables of interest, as outlined in section~\ref{sec:likelihood_computation}.
\begin{figure}
    \centering
    \includegraphics[width=1.00\linewidth]{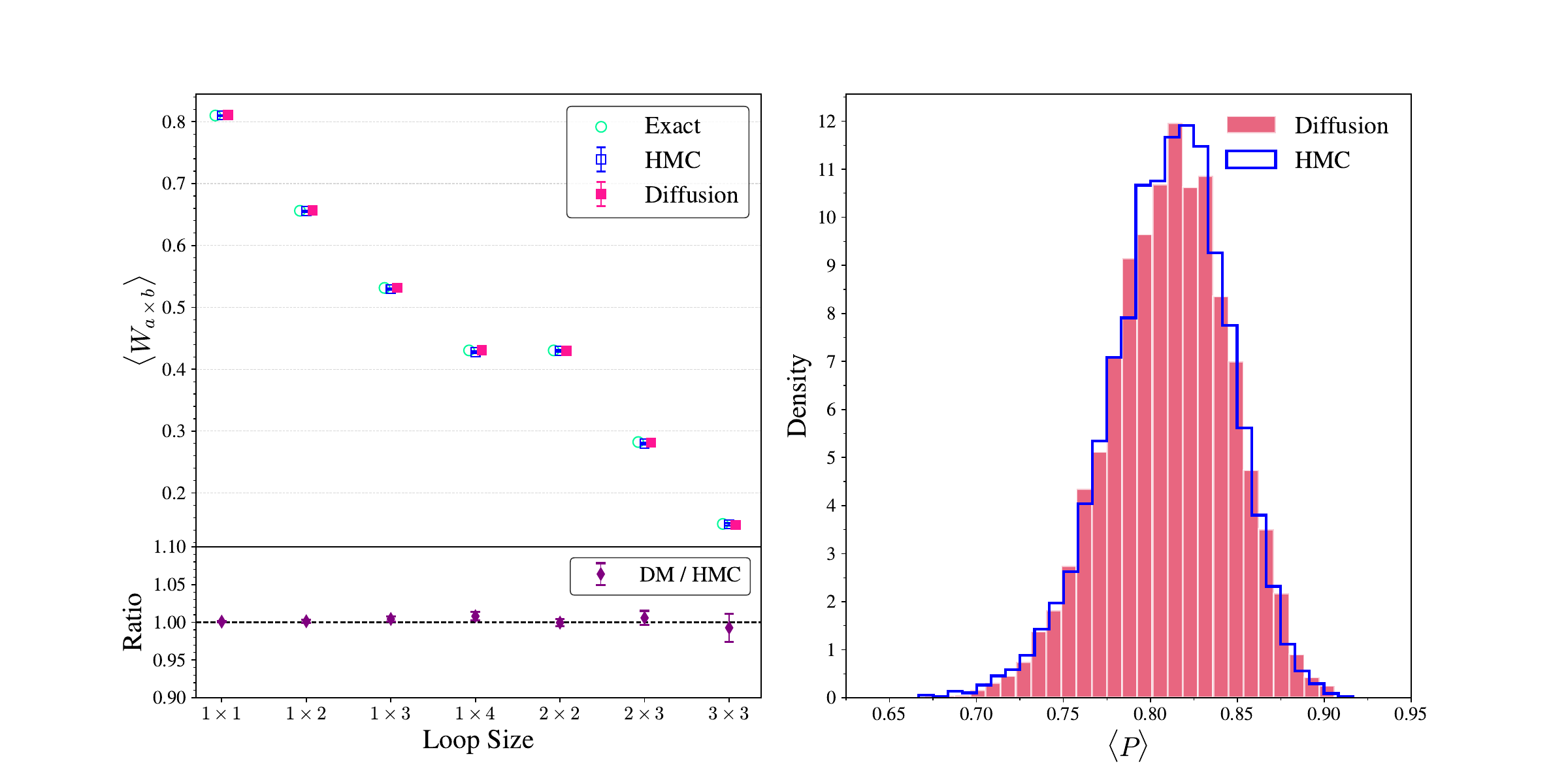}
    \caption{(Left) Measurements of the mean value and standard error of rectangular Wilson loops at different sizes $a \times b$, computed on 8192 configurations each from a trained diffusion model and an HMC test dataset for Abelian gauge theory at $\beta = 3.0$ on an $8 \times 8$ lattice. We also plot the exact values against the numerical estimates. (Right) normalized histograms of reweighted average plaquette values from the diffusion model and HMC test set.}
    \label{fig:wilson_loop_comparisons}
\end{figure}
We also compute \emph{topological susceptibility}, formally defined as
\begin{equation}
    \chi_Q := \frac{\langle Q^2\rangle}{V}.
\end{equation}
The topological susceptibility is associated with the fluctuations of the topological charge, and while its computation in lattice QCD poses many challenges~\cite{DelDebbio:2004ns,LuscherPalombi2010TopSus,Ce2015NonGaussian,Ce2016Topological,Athenodorou:2022TopSus}, the simple angular definition of $Q$ in equation~\eqref{eq:topo_charge} allows a straightforward computation of $\chi_Q$ in U(1) lattice gauge theory. 
Table~\ref{tab:observables_comparison} lists reweighted measurements of $\langle P\rangle$ and  $ \chi_Q$, which show good agreement between the data obtained with HMC and the diffusion model for smaller values of $\beta$. We observe some tension for the mean-value of $\langle P\rangle$ for larger values of $\beta$. 
This might be due to error in estimating the Jacobian of the induced transformation, though we leave this investigation to future work.

\sisetup{
  uncertainty-mode = compact,
  separate-uncertainty = false,
  detect-weight = true,
  detect-inline-weight = math,
}

\newcommand{\NA}{\multicolumn{1}{c}{--}}

\begin{table}
\centering
\setlength{\tabcolsep}{8pt}
\renewcommand{\arraystretch}{1.1}

\begin{tabular}{
    c l
    S[table-format=1.3(2)] S[table-format=1.4(2)] c
    S[table-format=1.3(2)] S[table-format=1.4(2)] c
}
    \specialrule{.12em}{0pt}{0pt}

    %\rowcolor{gray!10}
    & &\multicolumn{3}{c}{$L = 8$} & \multicolumn{3}{c}{$L = 16$} \\

    \cmidrule[0.1pt](lr){3-5}
    \cmidrule[0.1pt](lr){6-8}

    %\rowcolor{gray!10}
    {$\beta$} & {Ensemble}
    & {$\langle P \rangle$} & {$\chi_Q$} & {ESS}
    & {$\langle P \rangle$} & {$\chi_Q$} & {ESS} \\

    \specialrule{.05em}{0pt}{0pt}

    \multirow{2}{*}{1.0}
        & HMC
        & 0.4460(8) & 0.0406(6) & --
        & 0.4463(4) & 0.0420(7) & -- \\
        & Diffusion
        & 0.4450(8) & 0.0406(6) & 93\%
        & 0.4463(4) & 0.0409(7) & 84\% \\

    \cmidrule(lr){2-8}

    \multirow{2}{*}{2.0}
        & HMC
        & 0.6971(5) & 0.0197(3) & --
        & 0.6977(1) & 0.0193(1) & -- \\
        & Diffusion
        & 0.6967(7) & 0.0194(3) & 81\%
        & 0.6970(4) & 0.0193(4) & 53\% \\

    \cmidrule(lr){2-8}

    \multirow{2}{*}{3.0}
        & HMC
        & 0.8098(4) & 0.0111(1) & --
        & 0.8099(1) & 0.0111(2) & -- \\
        & Diffusion
        & 0.8085(5) & 0.0110(1) & 63\%
        & 0.8098(4) & 0.0111(4) & 23\% \\

    \cmidrule(lr){2-8}

    \multirow{2}{*}{4.0}
        & HMC
        & 0.8632(3) & 0.0075(1) & --
        & 0.8637(1) & 0.0075(1) & -- \\
        & Diffusion
        & 0.8626(4) & 0.0075(2) & 46\%
        & 0.8627(2) & 0.0078(2) & 21\% \\

    \cmidrule(lr){2-8}

    \multirow{2}{*}{5.0}
        & HMC
        & 0.8935(2) & 0.0056(1) & --
        & 0.8932(1) & 0.0057(1) & -- \\
        & Diffusion
        & 0.8923(3) & 0.0057(2) & 41\%
        & 0.8930(3) & 0.0059(3) & 14\% \\

    \specialrule{.12em}{0pt}{0pt}
\end{tabular}
\caption{Comparison of observables for $8 \times 8$ and $16 \times 16$ lattices at couplings $\beta \in \{1, 2, 3, 4, 5\}$ between diffusion-generated and HMC ensembles. The HMC ensembles contain 8192 independent configurations, each generated with 200 thermalization trajectories. The diffusion ensembles also contain 8192 configurations, and their observables are reweighted.}
\label{tab:observables_comparison}
\end{table}
We obtain effective sample sizes in the range of 41\% to 93\% for the $8 \times 8$ ensembles and between 14\% and 84\% for the $16 \times 16$ ensembles in table~\ref{tab:observables_comparison}. As expected, the likelihood computation becomes more challenging with increasing lattice size; the ESS decreases with increasing volume, thereby reducing the efficiency of using the diffusion model to produce precise estimates for observables. In addition, we observe stronger downward scaling of the ESS with increasing $\beta$ in the $L=16$ regime as compared to $L=8$. In principle, this effect can be mitigated by using more expressive score networks or longer training times. As an alternative method, ref.~\cite{Zhu:2025pmw} applies MALA during the sampling phase for ${\rm U}(1)$ gauge theory. We reserve these ongoing investigations for future work.

\section{Conclusions and outlook}\label{sec:conclusion}
In this work, we develop symmetry-preserving diffusion models and subsequently test them on two-dimensional $\phi^4$ field theory and ${\rm U}(1)$ gauge theory. Through several numerical experiments, we demonstrate that our symmetry-preserving diffusion models produce ensembles yielding accurate physical observables. We show that our score networks correctly approximate the force fields and correctly learn to sample from the target densities, which is enhanced by enforcing symmetry. We employ a modified version of score matching that improves the quality of our models. It is noteworthy that score matching is a fast training procedure, which is particularly efficient in comparison to training continuous flow-based models that require the additional step of forward and backward passes through ODE solvers. We validate the proposed diffusion models by comparing several key observables with conventional sampling approaches.

We quantify our model's performance by computing effective sample sizes (ESS) and acceptance rates for resampled ensembles, where we find that our generated ensembles display significantly less autocorrelation after resampling compared to HMC data. As discussed, although this improvement comes at the cost of increased computational expense, it can be trivially parallelized. Most encouragingly, we illustrate that our paradigm, which combines force-regularization in training and encoding explicit symmetries, yields models that achieve stronger performance than non-equivariant baselines. While this is a promising start, to establish scalability, it is still necessary to examine the performance of this framework in higher dimensions and with much larger lattice sizes. 

Diffusion models are naturally amenable to different methods of sampling, several of which are detailed in ref.~\cite{Yang2023DiffusionSurvey}. Langevin dynamics has been applied to solve the reverse SDE~\cite{Dockhorn2022CriticallyDamped, Jolicoeur-Martineau2020AdversarialScoreMatching} with different implementations. Other modifications include SDE solvers with adaptive step sizes~\cite{Jolicoeur-Martineau2021GottaGoFast} and explicit predictor-corrector methods~\cite{Song2021ScoreSDE}. While we do not consider these additional sampling methods in this study, they form interesting and necessary avenues to explore in future work. Examining how different integrators, predictor-correctors, and other methods can impact the quality of generated samples is an important direction for future research on this topic.

Another natural future direction is to extend our framework to non-abelian gauge theories, especially ${\rm SU}(3)$, towards simulating QCD. The primary aim would be to design score networks that exactly respect ${\rm SU}(N)$ symmetries, which can be achieved, for example, through the lattice gauge-equivariant convolutional neural networks (L-CNNs) introduced in ref.~\cite{Favoni:2020reg}. Integrating fermions into these models should be straightforward because pre-generated gauge configurations will already include all information about fermionic determinants. Nevertheless, the score function may need to be more expressive due to the non-locality of the fermion determinant. It would also be interesting to scrutinize the $\phi^4$ theory more deeply in this context. In particular, diffusion models could be used to characterize the phase transitions and criticality of the theory. This also creates a strong opportunity for testing more advanced score network architectures. Additionally, developing a better understanding of how hyperparameters, for both training and inference, influence sample quality would be extremely useful for scaling models upwards. 

To echo the conclusions of ref.~\cite{Abbott:2022zsh}, we emphasize that diffusion models, like normalizing flows, each represent just one approach in a much broader taxonomy of next-generation deep generative models and ML-enhanced sampling techniques for LQFT. Developing a fuller perspective on the future capabilities of the methods proposed here is worth exploring.

\acknowledgments
The authors are thankful to Gert Aarts and Lingxiao Wang for correspondence and helpful comments, and Antoine Misery for his work in the early stages of the project.
OV and AXK acknowledge the ETH Zürich Institute for Theoretical Physics and the Pauli Center for Theoretical Studies for support and hospitality in summers 2024 and 2025, as well as the Mainz Institute for Theoretical Physics (MITP) of the Cluster of Excellence PRISMA$^+$ (Project ID 390831469) for its hospitality and support.
AXK and MKM thank the Kavli Institute for Theoretical Physics (KITP) for hospitality and support during the program “What is Particle Theory?” The KITP is supported in part by the National Science Foundation under Grant No. PHY-2309135. MKM is grateful for the hospitality of Perimeter Institute where part of this work was carried out. Research at Perimeter Institute is supported in part by the Government of Canada through the Department of
Innovation, Science and Economic Development and by the Province of Ontario through the Ministry of Colleges and Universities. This research was also supported in part by the Simons Foundation through the Simons Foundation Emmy Noether Fellows Program at Perimeter Institute. AXK and OV acknowledge support from the U.S. Department  of Energy, Office of Science under grant Contract Number DE-SC0015655. JK and MKM acknowledge the support received from the Horizon Europe project interTwin, funded by the European Union Grant Agreement Number 101058386. OV acknowledges support from a University of Illinois Graduate College Fellowship, as well as from a Sloan Scholarship through the Alfred P. Sloan Foundation’s University Center of Exemplary Mentoring, awarded in 2023-2024. Numerical experiments and data analysis were carried out with PyTorch~\cite{paszke2019pytorch}, NumPy~\cite{Harris2020ArrayProgramming}, and SciPy~\cite{Virtanen2020SciPy1}. Figures were produced using Matplotlib~\cite{Hunter2007Matplotlib}. Almost all computational work in this project was performed on a single NVIDIA TITAN Xp GPU through the ETH D-PHYS \texttt{spaceml4} compute node.
Some of the experiments in this research used the DeltaAI advanced computing and data resource, which is supported by the National Science Foundation (award OAC 2320345) and the State of Illinois. DeltaAI is a joint effort of the University of Illinois Urbana-Champaign and its National Center for Supercomputing Applications.
Access to DeltaAI was provided in part by the Illinois Computes project which is supported by the University of Illinois Urbana-Champaign.

\appendix

\section{Lattice field theory formalism}\label{apx:lqft_formalism}
Here we present further details on the formulation of relevant aspects of LQFT which supplement the intuitions presented in the main body of our work.

\subsection{Pure abelian gauge theory on the lattice}\label{apx:formulation_u1_gauge_theory}
A ${\rm U}(1)$ gauge theory in continuous Minkowski spacetime can be defined by the Maxwell Lagrangian density:
\begin{equation}
    \mathscr{L} = -\frac{1}{4}F_{\mu\nu}F^{\mu\nu},
\end{equation}
where the components of the abelian field strength tensor $F_{\mu\nu}$ are given by
\begin{equation}
    F_{\mu\nu} = \partial_\mu A_\nu - \partial_\nu A_\mu.
\end{equation}
This Lagrangian is invariant under local ${\rm U}(1)$ transformations, represented as translations of the gauge field by a gradient:
\begin{equation}\label{eq:continuum_u1_gauge_transform}
    A_\mu(x) \mapsto A_\mu(x) - \partial_\mu \alpha(x)
\end{equation}
where $\alpha$ is an arbitrary scalar-valued function. In the defining, one-dimensional representation of the Lie group ${\rm U}(1)$, the gauge degrees of freedom $A_\mu \in \mathfrak{u}(1)$ are algebra-valued spacetime vector fields which play the role of affine connections for the gauge group. When discretizing the pure U(1) theory onto the lattice, one replaces the algebra-valued gauge fields $A_\mu$ with group-valued parallel transporters $U_\mu$, called \textit{gauge links}, defined as path-ordered Wilson lines between infinitesimally separated adjacent lattice sites:
\begin{equation}
	U_\mu(x) := \mathcal{P} \exp\left\{i \int_{x}^{x + \hat{\mu}} A_\nu(y) \; dy^\nu \right\}.
\end{equation}
For sufficiently smoothly varying gauge fields, the limit of small lattice spacing allows us to replace $A_\mu$ by an approximately constant phase parameter $\theta_\mu \in [-\pi, \pi)$. Then the gauge links can simply be written as complex phases $U \in \mathbb{C}$:
\begin{equation}\label{eq:gauge_link_phase_angle}
    U_\mu(x) = e^{i \theta_\mu(x)}.
\end{equation}
Under arbitrary gauge transformations $\Omega(x) \in {\rm U}(1)$, the links transform via
\begin{equation}
    U_\mu(x) \mapsto \Omega(x) U_\mu(x) \Omega^*(x + \hat{\mu}).
\end{equation}
Gauge-invariant quantities on the lattice are therefore constructed as products of link variables around closed loops, since any closed loop that starts and ends at a lattice site $x$ will transform via conjugation by $\Omega(x)$. Since ${\rm U}(1)$ is abelian, gauge-invariance is trivial as the conjugation leaves the loop unchanged. As such, the gauge-invariant ${\rm U}(1)$ lattice action is written as
\begin{equation}
    S[U] = -\beta \sum_x \sum_{\mu < \nu} \Re P_{\mu\nu}(x),
\end{equation}
where $\beta$ corresponds to an inverse coupling strength and $P_{\mu\nu}$ are the smallest gauge-invariant units that can be constructed on the lattice, called \textit{plaquettes}, defined as
\begin{equation}
    P_{\mu\nu}(x) := U_\mu(x) U_\nu(x + \hat{\mu}) U_\mu^*(x + \hat{\nu}) U_\nu^*(x).
\end{equation}
Since both ${\rm U}(1)$ gauge links and gauge transformations are elements of a one-parameter compact Lie group, we may adopt the one-dimensional \textit{angular representation} in which we parameterize the transformations as complex phases as well:
\begin{equation}\label{eq:gauge_transf_phase_angle}
	\Omega(x) = e^{i \omega(x)}
\end{equation}
such that the gauge transformation of link variables can be fully specified in the Lie algebra as
\begin{equation}\label{eq:angular_gauge_transform}
    \theta_\mu(x) \mapsto \theta_\mu(x) + \omega(x) - \omega(x + \hat{\mu}),
\end{equation}
which can be seen as a discretized version of the continuous Abelian gauge field transformation~\eqref{eq:continuum_u1_gauge_transform}. Furthermore, plaquettes can be written more simply in the angular representation as
\begin{equation}\label{eq:plaquette_angle}
    \phi_{\mu\nu}(x) := \theta_\mu(x) + \theta_\nu(x + \hat{\mu}) - \theta_\mu(x + \hat{\nu}) - \theta_\nu(x),
\end{equation}
where we make the identification that $P_{\mu\nu} = e^{i \phi_{\mu\nu}}$. Furthermore, one can easily show that $\phi_{\mu\nu}(x)$ is invariant under transformations of the form \eqref{eq:angular_gauge_transform}, so gauge-invariance is preserved in the angular representation. Having parameterized both gauge links $\theta$ and gauge transformations $\Omega$ as real-valued phases in the compact algebra $\omega \in [-\pi, \pi] \cong \mathfrak{u}(1)$, we may write down the ${\rm U}(1)$ lattice action:
\begin{equation}\label{eq:u1_gauge_action}
    S[\theta] = -\beta \sum_{x } \sum_{\mu < \nu} \cos \phi_{\mu\nu}(x),
\end{equation}
where one immediately sees that this action must be gauge-invariant since it is written entirely in terms of gauge-invariant plaquette phases.

\subsection{Calculation of the ${\rm U(1)}$ gauge force}\label{apx:u1_score_func}
In this section we present an analytical derivation of the exact ${\rm U(1)}$ score function, i.e. the force. For further reading and comparison of results, we direct readers to section 8.2.3 of ref.~\cite{Gattringer:2010zz}, where the ${\rm SU}(3)$ force is similarly computed for use in the HMC algorithm. 

In the numerical experiments presented in section~\ref{sec:u1_2d_application}, the relevant data are the link angles $\theta_\mu(x)$ corresponding to the gauge field configurations. To compute the force, we must differentiate eq.~\eqref{eq:u1_gauge_action} with respect to a fixed angle configuration $\theta_\alpha(y)$ for some arbitrary direction $\alpha$ and some lattice site $y$. The chain rule gives
\begin{equation}
    \boldsymbol{s}(\theta_\alpha(y)) = -\frac{\partial}{\partial \theta_\alpha(y)} S[\theta] = -\beta\sum_x \sum_{\mu < \nu} \sin\phi_{\mu\nu}(x) \frac{\partial \phi_{\mu\nu}(x)}{\partial \theta_\alpha (y)}.
\end{equation}
Then, differentiating the plaquette angle~\eqref{eq:plaquette_angle} yields
\begin{align}
	\frac{\partial \phi_{\mu\nu}(x)}{\partial \theta_\alpha (y)} &= \delta_{x, y} \delta_{\mu\alpha} + \delta_{x + \hat{\mu}, y}\delta_{\nu\alpha} - \delta_{x +\hat{\nu}, y}\delta_{\mu\alpha} - \delta_{x, y}\delta_{\nu\alpha} \\
	&= \delta_{\mu\alpha} \left(\delta_{x, y} - \delta_{x +\hat{\nu}, y}\right) + \delta_{\nu\alpha}\left(\delta_{x + \hat{\mu}, y} - \delta_{x, y} \right).
\end{align}
The plaquette angle satisfies two useful properties for all lattice sites $x$. Firstly, if the directions are the same, the angle vanishes: $\phi_{\mu\mu}(x) = 0$. Secondly, under the transposition $\mu \leftrightarrow \nu$, the plaquette is negated: $\phi_{\mu\nu}(x) = -\phi_{\nu\mu}(x)$.
Using these facts, the double sum over spacetime directions can be replaced by $\sum_{\mu < \nu} \rightarrow \tfrac{1}{2}\sum_{\mu, \nu}$.
We now have
\begin{align}
	\boldsymbol{s}(\theta_\alpha(y)) &= -\frac{\beta}{2}\sum_x \sum_\nu\left(\delta_{x, y} - \delta_{x +\hat{\nu}, y}\right)\sin\phi_{\alpha\nu}(x) \nonumber \\ 
	&\;\;\;\;\; - \frac{\beta}{2}\sum_x \sum_\mu \left(\delta_{x + \hat{\mu}, y} - \delta_{x, y} \right) \sin\phi_{\mu\alpha}(x),
\end{align}
where we used the Kronecker deltas to complete the inner sums over $\mu$ and $\nu$ in the first and second terms, respectively. Summing over the lattice sites gives 
\begin{align}
	\boldsymbol{s}(\theta_\alpha(y)) &= -\frac{\beta}{2} \sum_\nu \left(\sin\phi_{\alpha\nu}(y) - \sin\phi_{\alpha\nu}(y - \hat{\nu})\right) \nonumber \\ 
	&\;\;\;\;\; - \frac{\beta}{2} \sum_\mu \left(\sin\phi_{\mu\alpha}(y - \hat{\mu}) - \sin\phi_{\mu\alpha}(y) \right).
\end{align}
Relabeling the summed indices $\mu$ and $\nu$ to a common index $\gamma$ and using the antisymmetry of the angular plaquette gives the final result:
\begin{equation}
	\boldsymbol{s}(\theta_\alpha(y)) = -\beta \sum_\gamma \left[\sin\phi_{\gamma\alpha}(y) - \sin\phi_{\gamma\alpha}(y - \hat{\gamma})\right].
\end{equation}
Upon inspection, it becomes clear that $\boldsymbol{s}(\theta)$ is ${\rm U}(1)$-invariant in the angular representation, since it is built entirely out of gauge-invariant plaquette phases $\phi_{\mu\nu}$. 

\section{Stochastic dynamics of scalar fields}\label{apx:fp_dynamics}
\subsection{Derivation of the Fokker-Planck equation}\label{apx:fp_equation}
In this section, we derive the Fokker-Planck equation (FPE) for scalar fields following ref.~\cite{Zinn-Justin:1989rgp}.
This equation describes the time evolution of the probability density function (PDF) $p_t(\phi_t)$. We denote a scalar field by $\phi_t(x)$, where $x$ corresponds to a lattice site and $t \geq 0$ indexes the fictitious diffusion time. For brevity, we will suppress the position dependence and use the shorthand $\phi_t$. The forward diffusion process can be expressed as the SDE:
\begin{equation}
   \frac{d \phi_t}{dt} = f(\phi_t) + \sigma \eta_t,
   \label{eq:diffusion:SDE:phi_t}
\end{equation}
where $f(\phi_t)$ represents a deterministic drift term, $\eta_t$ is a stochastic term that we set to the standard white Gaussian noise,
and $\sigma$ is constant with respect to the scalar field, but it can depend on time. For a small step size $h$ in time, we have
\begin{equation}
    \phi_{t + h} = \phi_t + h f(\phi_t) + \sqrt{h} \sigma \int_t^{t+h} \frac{d\tau}{\sqrt{h}} \; \eta_\tau,
\end{equation}
where we neglected all terms of size $O(h^{3/2})$. The integral yields another noise term with zero mean and unit variance.
First, we calculate the conditional PDF of $\phi_{t + h} = \varphi$
at time $t + h$ subject to $\phi_t = \varphi_0$ at time $t$:
\begin{align}
    p_{t + h|t}(\varphi | \varphi_0)
    &= \mathbb{E}_{\eta}\left[
    \delta\left(\varphi - \varphi_0 - h f(\varphi_0) - \sqrt{h} \sigma \eta\right)
    \right] \nonumber
    \\
    &= \int \frac{d\zeta}{(2\pi)^n} \; e^{i \zeta \cdot (\varphi - \varphi_0 - h f(\varphi_0))}
    \mathbb{E}_{\eta} \left[e^{-i \zeta \cdot \sqrt{h} \sigma \eta}\right]
    \nonumber
    \\
    &= \int \frac{d\zeta}{(2\pi)^n} \; e^{i \zeta \cdot (\varphi - \varphi_0 - h f(\varphi_0))}
    e^{-\frac{h\sigma^2}{2} \zeta\cdot \zeta}\, .
    \label{eq:conitional:phi_t}
\end{align}
Here, we treat $\varphi$ as a vector of dimension $n$, and the delta function is understood as a product of $n$ delta functions over the components of the vector. In the second line, we used the Fourier representation of the delta function by introducing $\zeta$ as a vector of the same size of $\varphi$. Then, $p_{t+h}(\varphi)$ can be expressed as
\begin{align}
   p_{t+h}(\varphi) &=\int d \varphi_0\, p_{t + h|t}(\varphi | \varphi_0) p_t(\varphi_0) \nonumber\\
   &=\int d \varphi_0 \int \frac{d\zeta}{(2\pi)^n} \; e^{i \zeta \cdot (\varphi - \varphi_0 - h f(\varphi_0))}
    e^{-\frac{h\sigma^2}{2} \zeta\cdot \zeta}\,
    p_t(\varphi_0)\,.
\end{align}
The time derivative of $p_{t}(\phi_t)$ then reads
\begin{align}
    \frac{\partial}{\partial t} p_{t}(\varphi) 
    &= \lim_{h\to 0} \frac{p_{t + h}(\varphi) - p_{t}(\varphi)}{h}
    \nonumber
    \\
    &= \int d \varphi_0 \; \lim_{h\to 0}
    \frac{p_{t + h|t}(\varphi | \varphi_0)
    - p_{t|t}(\varphi | \varphi_0)}{h} 
    p_t(\varphi_0)
    \nonumber
    \\
    &= \int d \varphi_0 \; \int \frac{d\zeta}{(2\pi)^n} \; e^{i \zeta \cdot \varphi} e^{-i \zeta \cdot \varphi_0}
    \left(- i \zeta \cdot f(\varphi_0) - \frac{\sigma^2}{2} \zeta\cdot \zeta\right) 
    p_t(\varphi_0)
    \nonumber
    \\
    &= \int d \varphi_0 \; \delta(\varphi - \varphi_0)
    \left(- \nabla \cdot f(\varphi_0)
    + \frac{\sigma^2}{2} \nabla \cdot \nabla \right) 
    p_t(\varphi_0)
    \nonumber
    \\
    &= \nabla \cdot \left(- f(\varphi)
    + \frac{\sigma^2}{2} \nabla \right)
    p_t(\varphi)
    \label{eq:FP:phi_t}.
\end{align}
In deriving the fourth line, we applied integration by parts and assumed that the integrand vanishes at the boundaries. This concludes our derivation of the Fokker-Planck equation.

\subsection{Reverse Langevin process}\label{apx:reverse_process_derivation}
We now discuss the reverse process in Langevin dynamics. To this end, we introduce a function $g$ and a SDE as
\begin{align}
   \frac{d \phi_t}{dt} &= g(\phi_t) + \tilde \sigma \eta_t
   \label{eq:diffusion:SDE:rev:phi_t}\,,
\end{align}
where $\eta_t$ is the standard white Gaussian noise
and $\tilde \sigma$ is constant with respect to the scalar field, but it can depend on time.
We demand equation~\eqref{eq:diffusion:SDE:rev:phi_t} to be the reverse of
equation~\eqref{eq:diffusion:SDE:phi_t} in the following sense: a sequential evolution of the random field $\phi_{t_1}$ governed by equation~\eqref{eq:diffusion:SDE:phi_t} from time $t_1$ to $t_2$ and by equation~\eqref{eq:diffusion:SDE:rev:phi_t} from $t_2$ to $t_1$ must not change
the PDF of the random field.
In other words, the FPEs of the forward and reverse processes must be
the same. We show below that this condition holds if we set
\begin{equation}
    g(\phi_t) = f(\phi_t) - \frac{\sigma^2 + \tilde \sigma^2}{2} \nabla \log p_t(\phi_t).
    \label{eq:set-reverse-g}
\end{equation}
We now present the proof. We use $p_t(\phi_t)$ and $q_t(\phi_t)$ to denote the PDF of $\phi_t$ in the forward and reverse processes, respectively. The aim is to obtain the function $g$ such that $q_t(\phi_t) = p_t(\phi_t)$. Having already derived the FPE for $p_t(\phi_t)$, we now derive the corresponding equation for $q_t(\phi_t)$. For a small $h$, equation~\eqref{eq:diffusion:SDE:rev:phi_t} can be solved backward in time as
\begin{equation}
    \phi_{t - h} = \phi_t - h g(\phi_t) + \sqrt{h} \tilde \sigma \int_t^{t-h} \frac{d\tau}{\sqrt{h}} \; \eta_\tau ,
\end{equation}
where we neglected all terms of size $O(h^{3/2})$. Similar to the forward process, the integral yields another normal noise with zero mean and unit variance independent of the overall sign. The conditional PDF of $\phi_{t - h} = \varphi$ at time $t - h$ subject to $\phi_t = \varphi_0$ at time $t$ reads
\begin{align}
    q_{t - h |t}(\varphi | \varphi_0)
    &= \mathbb{E}_{\eta}\left[
    \delta\left(\varphi - \varphi_0 + h g(\varphi_0) + \sqrt{h} \tilde \sigma \eta\right)
    \right] \nonumber
    \\
    &= \int \frac{d\zeta}{(2\pi)^n} \; e^{i \zeta \cdot (\varphi - \varphi_0 + h g(\varphi_0))}
    e^{-\frac{h\tilde \sigma^2}{2} \zeta\cdot \zeta}\, .
    \label{eq:conitional:rev:phi_t}
\end{align}
Note that there is a mismatch between the sign of the terms with $h$ in
equations~\eqref{eq:conitional:phi_t} and \eqref{eq:conitional:rev:phi_t}.
Because of the mismatch, the FPE of the reverse process looks different:
\begin{align}
    \frac{\partial}{\partial t} q_t(\varphi) 
    &= \lim_{h\to 0} \frac{q_{t - h}(\varphi) - q_t(\varphi)}{-h}
    \nonumber
    \\
    &= \nabla \cdot \left(- g(\varphi) - \frac{\tilde \sigma^2}{2}
    \nabla\right)
    q_t(\varphi)
    \label{eq:FP:rev:phi_t}.
\end{align}
Defining $g$ as in equation~\eqref{eq:set-reverse-g} removes the apparent difference and makes both FPEs identical (we assume there is a fixed time at which $p$ and $q$ are equal).

It is noteworthy that eq.~\eqref{eq:FP:rev:phi_t} defines a \emph{family} of reverse processes for the diffusion process, parameterized by $\tilde \sigma$. Remarkably, when $\tilde \sigma$ vanishes, the reverse process reduces to a deterministic differential equation.

\subsection{Solving the FPE along ODE trajectories}\label{apx:fp_ode_solution}
In this section, we provide a solution to the FPE and relate it to the probability flow ODE for samples in reverse diffusion time. We suppose that samples $\phi_t$ evolve in reverse time according to the deterministic process where $\tilde\sigma = 0$ in Eq.~\eqref{eq:diffusion:SDE:rev:phi_t}:
\begin{equation}
    \frac{d\phi_t}{dt} = f(\phi_t, t) - \frac{\sigma^2}{2} \nabla \log p_t(\phi_t).
\end{equation}
The total time derivative of the log density is then, by the chain rule:
\begin{equation}\label{eq:time_deriv_logp}
    \frac{d}{dt} \log p_t(\phi_t) = \frac{\partial}{\partial t}\log p_t(\phi_t) + \frac{d\phi_t}{dt} \cdot \nabla \log p_t(\phi_t).
\end{equation}
The FPE for $p_t$ is
\begin{equation}
    \partial_t p_t = -\nabla (fp_t) + \frac{\sigma^2}{2} \nabla^2 p_t,
\end{equation}
which yields
\begin{equation}
    \frac{\partial}{\partial t}\log p_t = \frac{1}{p_t}\left(-p_t\nabla \cdot f - f \cdot \nabla p_t + \frac{\sigma^2}{2}\nabla^2 p_t\right).
\end{equation}
Substituting into \eqref{eq:time_deriv_logp}, we get
\begin{align}
    \frac{d}{dt} \log p_t = -\nabla \cdot f -f \cdot \nabla \log p_t + \frac{\sigma^2}{2}\frac{\nabla^2 p_t}{p_t} \nonumber  + \frac{d\phi_t}{dt} \cdot \nabla \log p_t.
\end{align}
Substituting the deterministic flow for $\phi_t$ and replacing $\nabla \log p_t(\cdot)$ with $\boldsymbol{s}(\cdot, t)$, we have
\begin{align}
    \frac{d}{dt} \log p_t(\phi_t) &= -\nabla \cdot f -f \cdot \boldsymbol{s} + \frac{\sigma^2}{2} \left(\|\boldsymbol{s}\|^2 + \nabla \cdot \boldsymbol{s}\right) + \left(f - \frac{\sigma^2}{2} \boldsymbol{s} \right)\cdot \boldsymbol{s} \\
    &= -\nabla \cdot f(\phi_t, t) + \frac{\sigma^2}{2} \nabla \cdot \boldsymbol{s}(\phi_t, t).
\end{align}
Hence, the total derivative along the probability flow ODE trajectory is simply
\begin{equation}
    \frac{d}{dt} \log p_t(\phi_t) = -\nabla \cdot \left(f(\phi_t, t) - \frac{\sigma^2}{2}\boldsymbol{s}(\phi_t, t)\right),
\end{equation}
whose solution is obtained as
\begin{equation}
    p_t(\phi_t) = \exp\left[-\int_{t'}^t \nabla \cdot \left(f(\phi_\tau, \tau) - \frac{\sigma^2}{2}\boldsymbol{s}(\phi_\tau, \tau)\right) d\tau\right] p_{t'}(\phi_{t'}).
\end{equation}
The probability flow ODE gives a deterministic trajectory for samples $\phi_t$ that pushes probability mass, generating a flow of particles that matches the evolution of the density $p_t$. The pushforward of the starting distribution, $p_0$, along these trajectories yields the same marginal densities as the solution to the FPE.

\section{Equivariance of score functions}\label{apx:group_symmetries}
\subsection{Group symmetries and score functions}\label{apx:score_func_symmetries}
Here, we examine the relationship between the symmetries of probability distributions and the resulting symmetries of their corresponding score functions. Some notation: We denote groups by $\mathcal{G}$ and their algebras, when they exist, by $\mathfrak{g}$. Group elements are denoted $G \in \mathcal{G}$, and their representations by $\rho_G$. We use $\mathcal{F}$ to abstractly denote field space, i.e. the space where field configurations live.
\begin{definition}[Invariance]
    Let $\mathcal{G}$ be a group and let $f: \mathcal{X} \rightarrow \mathcal{X}$ be a function. Suppose an element $G \in \mathcal{G}$ acts on $x \in \mathcal{X}$ via $x \mapsto \rho_G(x)$, where $\rho_G$ is the representation of $G$. We say that $f$ is \textit{invariant} to $G$, or $G$-invariant, if
    \begin{equation}\label{eq:invariance}
        (f \circ\rho_G)(x) = f(x)
    \end{equation}
    for every $x$.  If eq.~\eqref{eq:invariance} holds for all $G \in \mathcal{G}$, then we say $f$ is invariant to $\mathcal{G}$, or $\mathcal{G}$-invariant. 
\end{definition}
\begin{definition}[Equivariance]
    Let $\mathcal{G}$ be a group and $f: \mathcal{X} \rightarrow \mathcal{Y}$ be a function. We say $f$ is \textit{equivariant} to $G \in \mathcal{G}$, or $G$-equivariant, if
    \begin{equation}\label{eq:equivariance}
        (f \circ \rho_G)(x) = (\rho_G \circ f)(x)
    \end{equation}
    for every $x$, i.e. if $f$ commutes with the group action. If instead $G$ acts on either of or both $\mathcal{X}$ and $\mathcal{Y}$, the definition extends to
    \begin{itemize}
        \item right-right equivariance:
        \begin{equation}\label{eq:right-right_equiv}
            f(x \cdot G) = f(x) \cdot G
        \end{equation}
        \item right-left equivariance:
        \begin{equation}\label{eq:right-left_equiv}
            f(x \cdot G) = G^{-1}\cdot f(x)
        \end{equation}
        \item left-right equivariance:
        \begin{equation}\label{eq:left-right_equiv}
            f(G \cdot x) = f(x) \cdot G^{-1}
        \end{equation}
    \end{itemize}
    If any of \eqref{eq:equivariance}--\eqref{eq:left-right_equiv} hold for all $G \in \mathcal{G}$, then we say $f$ is equivariant to $\mathcal{G}$, or $\mathcal{G}$-equivariant.
\end{definition}

\begin{definition}[Score Function]
    Given a time-dependent probability density $p_t(\phi)$ on $\mathcal{F}$, the \emph{score function} at time $t$ is the gradient of the log-density:
    \begin{equation}
        \boldsymbol{s}(\phi, t) := \nabla_\phi \log p_t(\phi) \in T^*_\phi \mathcal{F},
    \end{equation}
    where $T^*_\phi \mathcal{F}$ denotes the cotangent space at $\phi$.
\end{definition}
For simplicity, the present discussion focuses on the univariate case. We could generalize from a single lattice site to a full lattice by viewing relevant functions as being defined over product sets and product groups that cover the domain.
\begin{theorem}[Transformation Law of the Score Function]\label{thm:score_func_transformation}
    Let $G \in \mathcal{G}$ act smoothly on $\mathcal{F}$ by $\phi \mapsto \phi' = \rho_G(\phi)$. If $p_t$ is $\mathcal{G}$-invariant, then the score function transforms under the group action as
    \begin{equation}
        \boldsymbol{s}'(\phi', t) = \left(D_\phi \rho_G\right)^{-\top}\boldsymbol{s}(\phi, t),
    \end{equation}
    where the Jacobian (differential) $D_\phi \rho_G := \partial \rho_G(\phi) / \partial \phi$ represents the linearization of the group action at $\phi$.
\end{theorem}
\begin{proof}
    By the chain rule for gradients on manifolds and assuming $\rho_G$ is a diffeomorphism:
    \begin{align}
      \boldsymbol{s}'(\phi', t) = \nabla_{\phi'} \log p_t(\phi') = \left(\frac{\partial \phi}{\partial \phi'}\right)^\top \nabla_\phi \log p_t (\phi),
    \end{align}
    where we used invariance of the density: $p_t(\phi') = p_t(\phi)$. 
    Since $\rho_G$ is differentiable and invertible by assumption, we use the inverse function theorem to write
    \begin{equation}
        \frac{\partial \phi}{\partial \phi'} = \left( \frac{\partial \phi'}{\partial \phi} \right)^{-1},
    \end{equation}
    and thus the general transformation law for the score function becomes
    \begin{equation}\label{eq:score_fn_transformation}
    	\boldsymbol{s}(\phi, t) \mapsto \boldsymbol{s}'(\phi', t) = \left( \frac{\partial \rho_G(\phi)}{\partial \phi} \right)^{-\top} \boldsymbol{s}(\phi, t),
    \end{equation}
    as desired.
\end{proof}
Theorem~\ref{thm:score_func_transformation} shows that the score transforms like a covector under a change of variables induced by the symmetry group action. Though the result is unifying, its precise form depends on the nature of the field space and group action. In linear representations, for example, where $\mathcal{F}=\mathbb{R}^n$ with $\rho_G(\phi)=R_G\phi$, then $D_\phi \rho_G=R_G$ and the score transforms by
\begin{equation*}
    \boldsymbol{s}' = \left(R_G\right)^{-\top}\boldsymbol{s}
\end{equation*}
such as for $\mathcal{G}=\mathbb{Z}_2$ acting by sign flip, where $R_G=-\mathds{1}$ so $\boldsymbol{s}'=-\boldsymbol{s}$. If instead the fields take values in a Lie algebra, for instance, $\phi\in\mathfrak{g}$ and $\mathcal{G}$ acts by the adjoint representation $\rho_G(\phi) = \text{Ad}_G \phi$, then $D_\phi \rho_G=\text{Ad}_G$, so
\begin{equation*}
    \boldsymbol{s}' = \left({\rm Ad}_G\right)^{-\top}\boldsymbol{s}.
\end{equation*}
Likewise, if the fields $\phi=U\in\mathcal{G}$ are themselves group-valued and transform by conjugation $\rho_G(U)=GUG^{-1}$, then $D_U\rho_G$ gives the Jacobian acting on the tangent space $T_U\mathcal{G}$, and the score transforms via the dual action,
\begin{equation*}
    \boldsymbol{s}'(U') = \left(D_U \rho_G\right)^{-\top} \boldsymbol{s}(U)
\end{equation*}
reflecting covariance under the group’s induced action on the tangent space.

\subsection{Discussion on the ${\rm U}(1)$ symmetry of score functions}\label{apx:u1_symmetry_details}
Now, we extend our previous discussion on the representation-dependence of score function transformations to the  specific case of the symmetries of the ${\rm U}(1)$ score function.

We first work directly on the gauge group. Denote the score function at a lattice site $x$, viewed as a function of the gauge link $U_\mu$ for each $\mu$, as
\begin{equation}
    \boldsymbol{s}_\mu(x) = -\frac{\partial S[U]}{\partial U_\mu(x)}.
\end{equation}
In the group representation, link variables transform under the bi-regular representation of ${\rm U}(1)$ via
\begin{equation}
    U_\mu(x) \mapsto U_\mu'(x) = \Omega(x) U_\mu(x)\Omega^\dagger(x + \hat{\mu}).
\end{equation}
In turn, the score function for the transformed variable can be computed using the chain rule:
\begin{equation}
    \boldsymbol{s}_\mu'(x) = - \frac{\partial S[U']}{\partial U_\mu'(x)} =  - \frac{\partial U_\mu(x)}{\partial U_\mu'(x)} \frac{\partial S[U]}{\partial U_\mu(x)},
\end{equation}
where we assumed the action is gauge-invariant, i.e. $S[U']~=~S[U]$. Since
\begin{equation}
    U_\mu'(x)~=~\Omega^\dagger(x) U_\mu(x) \Omega(x + \hat{\mu}),    
\end{equation}
we can easily compute the Jacobian:
\begin{equation}
    \frac{\partial U_\mu(x)}{\partial U_\mu'(x)} = \Omega^\dagger(x) \Omega(x + \hat{\mu}). 
\end{equation}
Then, since ${\rm U}(1)$ is Abelian, we can write
\begin{equation}
    \boldsymbol{s}_\mu'(x) = \Omega^\dagger(x) \boldsymbol{s}_\mu(x) \Omega(x + \hat{\mu}).
\end{equation}
Hence, in the group space, the score function transforms equivariantly under the bi-regular transformation.

Now we turn to the Lie algebra $\mathfrak{u}(1) \cong \mathbb{R}$ and consider the same argument in terms of angular variables, where the links and gauge transformations adopt their angular definitions as in equations~\eqref{eq:gauge_link_phase_angle} and \eqref{eq:gauge_transf_phase_angle}, respectively. We view the score as a function of the real valued link phases $\theta_\mu(x) \in [-\pi, \pi)$ and denote
\begin{equation}
    \boldsymbol{s}_\mu(x) := -\frac{\partial S[\theta]}{\partial \theta_\mu(x)}.
\end{equation}
The score function transformation law is
\begin{equation}
    \boldsymbol{s}_\mu'(x) = \frac{\partial \theta_\mu(x)}{\partial \theta_\mu'(x)}\boldsymbol{s}_\mu(x).
\end{equation}
The link phases transform under the group according to the linear shift given in eq.~\eqref{eq:angular_gauge_transform}, so the group action in the algebra is by affine translation. Thus, the Jacobian is simply the identity, meaning that, when expressed as a function of angular variables in the Lie algebra, the score is gauge-invariant. In other words, the score function transforms under the adjoint representation, which is trivial for ${\rm U}(1)$.

\paragraph{Data Availability Statement.} This article has associated data in a data repository. Pedagogical, tutorial-style code notebooks that walk through simplified versions of the examples presented in sections \ref{sec:0d_toy_example}, \ref{sec:phi4_application}, and \ref{sec:u1_2d_application} can be found at \url{https://gitlab.com/ovega141/diffusion\_for\_lqft}.

\paragraph{Code Availability Statement.} This article has associated code in a code repository. Pedagogical, tutorial-style code notebooks that walk through simplified versions of the examples presented in sections \ref{sec:0d_toy_example}, \ref{sec:phi4_application}, and \ref{sec:u1_2d_application} can be found at \url{https://gitlab.com/ovega141/diffusion\_for\_lqft}.

\paragraph{Open Access.} This article is distributed under the terms of the Creative Commons Attribution License (\href{https://creativecommons.org/licenses/by/4.0/}{CC-BY4.0}), which permits any use, distribution and reproduction in any medium, provided the original author(s) and source are credited.

% Bibliography
\bibliographystyle{JHEP}
\bibliography{refs.bib}

\end{document}